\documentclass[10pt,journal,compsoc]{IEEEtran}
\pdfoutput=1
\usepackage{bm,amsmath,amssymb,mathrsfs,amsfonts,amsthm,url}
\usepackage{graphicx,booktabs,bbding}
\usepackage{algorithm}
\usepackage{algpseudocode}
\usepackage[colorlinks,linkcolor=black]{hyperref}

\newtheorem{theorem}{Theorem}
\UseRawInputEncoding
\hypersetup{colorlinks,
	linkcolor=black,%
	citecolor=black}
\ifCLASSOPTIONcompsoc
  \usepackage[nocompress]{cite}
\else
  \usepackage{cite}
\fi

\ifCLASSINFOpdf

\else

\fi

\begin{document}

\title{Self-Supervised Deep Graph Embedding with High-Order Information Fusion for Community Discovery}

\author{Shuliang~Xu \href{https://orcid.org/0000-0002-5464-2354}{\includegraphics[scale=0.25]{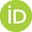}},
        Shenglan~Liu \href{https://orcid.org/0000-0003-3823-4200}{\includegraphics[scale=0.25]{orcid.png}},~\IEEEmembership{Member,~IEEE,}
        and~Lin~Feng \href{https://orcid.org/0000-0002-4942-2293}{\includegraphics[scale=0.25]{orcid.png}}
\IEEEcompsocitemizethanks{\IEEEcompsocthanksitem S. Xu is with the Faculty of Electronic Information and Electrical Engineering, Dalian University of Technology, Dalian, 116024 China.\protect\\
E-mail: slx\_cs@mail.dlut.edu.cn
\IEEEcompsocthanksitem S. Liu and L. Feng are with the School of Innovation and Entrepreneurship, Dalian University of Technology, Dalian, 116024 China.\protect\\
E-mail: liusl@dlut.edu.cn; fenglin@dlut.edu.cn}
\thanks{Manuscript received XX XX, XXXX; revised XX XX, XXXX.}}

\markboth{Journal of \LaTeX\ Class Files,~Vol.~14, No.~8, August~2015}%
{Shell \MakeLowercase{\textit{et al.}}: Bare Demo of IEEEtran.cls for Computer Society Journals}

\IEEEtitleabstractindextext{%
\begin{abstract}
Deep graph embedding is an important approach for community discovery. Deep graph neural network with self-supervised mechanism can obtain the low-dimensional embedding vectors of nodes from unlabeled and unstructured graph data. The high-order information of graph can provide more abundant structure information for the representation learning of nodes. However, most self-supervised graph neural networks only use adjacency matrix as the input topology information of graph and cannot obtain too high-order information since the number of layers of graph neural network is fairly limited. If there are too many layers, the phenomenon of over smoothing will appear. Therefore how to obtain and fuse high-order information of graph by a shallow graph neural network is an important problem. In this paper, a deep graph embedding algorithm with self-supervised mechanism for community discovery is proposed. The proposed algorithm uses self-supervised mechanism and different high-order information of graph to train multiple deep graph convolution neural networks. The outputs of multiple graph convolution neural networks are fused to extract the representations of nodes which include the attribute and structure information of a graph. In addition, data augmentation and negative sampling are introduced into the training process to facilitate the improvement of embedding result. The proposed algorithm and the comparison algorithms are conducted on the five experimental data sets. The experimental results show that the proposed algorithm outperforms the comparison algorithms on the most experimental data sets. The experimental results demonstrate that the proposed algorithm is an effective algorithm for community discovery.
\end{abstract}

}

\maketitle

\IEEEdisplaynontitleabstractindextext

\IEEEpeerreviewmaketitle

\section{Introduction}

\IEEEPARstart{N}{owadays}, graph is more and more common in a wide range of applications, such as social network analysis \cite{tang2017computational,liu2020deep}, citation network analysis \cite{yang2015defining}, product recommendation system \cite{du2019sequential,cen2020controllable}, knowledge graph \cite{zhang2019iteratively} and protein-protein interaction \cite{liu2019integrating}, etc. It is different from the traditional structured data that graph data is unstructured. The high dimension, unstructure and sparsity of graph data bring great challenges for community discovery. Therefore it is significant to transform graph data from the high dimensional and sparse space into the low dimensional and dense subspace.

Graph embedding is to learn a linear or nonlinear map function that can project graph data into low dimensional and dense vector space to facilitate downstream tasks such as product recommendation, community discovery, etc \cite{cui2018survey}. However, the conventional embedding algorithms, such as Locally Linear Embedding (LLE) \cite{roweis2000nonlinear}, Laplacian Eigenmap (LE) \cite{belkin2002laplacian} and Non-Negative Matrix Factorization (NMF) \cite{lee1999learning}, etc. only focus on the local structure of graph and ignore the higher-order structure information and attribute information. Recently, with the advances in deep learning, deep graph convolution neural network is an effective approach for community discovery \cite{rong2020deep} and it can obtain a strong representation ability by stacking multiple hidden layers. Deep graph convolution neural network can automatically learn the low dimensional and dense vectors of nodes from unstructured graph data and it can also colligate the structure information and attribute information into the embedding result. Therefore there are many advantages for deep graph convolution neural network dealing with graph data and researchers have paid much attention to this field.

Graph embedding is a hot topic in graph mining field. Up to now, researchers propose many graph embedding algorithms. For those algorithms, they are mainly divided into three categories: matrix factorization, random walk and graph neural network.

Matrix factorization is to decompose the adjacency matrix of a graph or the interactor matrix related to adjacency matrix. The representative algorithms are NetSMF \cite{qiu2019netsmf}, FONPE \cite{pang2017flexible}, AROPE \cite{zhang2018arbitrary}, MNMF \cite{wang2017community}, NSP \cite{qiu2019noise}, etc. However, there is a high time complexity for the most matrix factorization algorithms although the decomposed result has a good explanation. It means the scalability of the matrix factorization algorithms are not good. Therefore they are not suitable for community discovery tasks when the scale of the network is relatively large.

Random walk is to generate a random walk path sequence or a local subgraph as the context for each node and then it uses neural language model, such as Word2Vec \cite{mikolov2013efficient}, etc., to generate the low dimensional embedding vectors. The representative algorithms include DeepWalk \cite{perozzi2014deepwalk}, LINE \cite{tang2015line}, struc2vec \cite{ribeiro2017struc2vec}, metapath2vec \cite{dong2017metapath2vec}, SEED \cite{wang2019inductive}, SPINE \cite{guo2019spine}, MRF \cite{jin2019incorporating}, etc. However, the path of random walk depends on the topology information of a graph and cannot integrate the attribute information of nodes which is common in attribute graph. In addition, it is not possible to intervene in walk path according to the characteristics of a graph and random walk is biased to the nodes with large degree. Therefore the shortcomings of random walk place restrictions on its applications.

Graph neural network is an effective approach for graph embedding \cite{kipf2016semi,wu2020comprehensive} and it is proposed for dealing with unstructured graph data. Graph neural network can update the features of  nodes from their neighbors. With the powerful representation ability of deep learning, deep graph neural network achieves great success in graph mining although it may suffer from over smoothing problem \cite{xu2020powerful,loukas2020graph,chen2020simple}. The representative graph neural network algorithms are GraphSAGE \cite{hamilton2017inductive}, DropEdge \cite{rong2020dropedge}, GCKN \cite{chen2020convolutional}, GMI \cite{peng2020graph} and SDCN \cite{bo2020structural}, etc. However, how to efficiently train a graph neural network is still an open problem. Beyond that, the over smoothing problem restricts the depth of graph neural network although DropEdge and Dropcluster \cite{zhangdropping2020} propose some measures to slow down over smoothing. Nevertheless, DropEdge and Dropcluster cannot eliminate over smoothing and the performance of graph neural network still degenerates after the number layer of layer is too large. Therefore most graph neural networks are shallow layer and a shallow hidden layer restrains graph neural network from obtaining higher level features of nodes.

The high-order structure information of graph and the attribute information of nodes play important roles on graph embedding which can provide more information for community discovery. At present, due to the limitation of the number of hidden layers of graph neural network, most graph neural networks cannot obtain too high level embedding features of nodes. Therefore a self-supervised deep graph embedding algorithm with high-order information fusion for community discovery (SDGE) is proposed in this paper. SDGE uses multiple graph convolution neural networks which can integrate the attribute and structure information by convolution and aggregation operator to learn the dimensional embedding vectors. Deep graph convolution neural networks are trained from different high-order information and attribute information of nodes by introducing self-supervised learning mechanism. The final result is the fusion of the outputs of  multiple graph convolution neural networks. The main contributions of this paper are as follows:
\begin{itemize}
  \item A deep graph embedding algorithm is proposed in this paper. The different high-order information is employed to train multiple graph convolution neural networks. The final result is the output fusion of the multiple graph convolution neural networks.
  \item A data augmentation approach and negative sampling mechanism are introduced to improve the performance of the proposed algorithm. The graph convolution neural networks can be trained by self-supervised learning mechanism.
  \item The proposed algorithm introduces spectral propagation to enhance the embedding result.
  \item The proposed algorithm can keep the structure and attribute similarity in the low dimensional embedding space. The experimental results demonstrate the effectiveness of the proposed algorithm.
\end{itemize}

The rest of this paper is organized as follows: Section 2 reviews some related works; Sections 3 describes the theory and the detailed steps of the proposed algorithm; the experimental results and analysis are presented in Section 4; Section 5 concludes the paper and gives some research directions in the future.

\section{Related works}
Graph embedding is an important topic in graph mining and it is closely related to community discovery. In the early days, most graph embedding algorithms are from matrix factorization and random walk. In recent years, graph neural network facilitates the developments of this field and many deep graph embedding algorithms are proposed.

Wang et al. propose a deep attentional embedding approach \cite{wang2019attributed} called as DAEGC. DAEGC is a development of graph convolution neural network \cite{kipf2016semi}. It utilizes autocoder with double hidden layers to learn the embedding vectors of nodes, and then uses the clustering result to adjust the parameters of hidden layers by self-supervised learning. In graph neural network, DAEGC considers the weights of nodes and introduces attention mechanism to restructure the adjacency matrix and learn the target distribution.

Fan et al. propose a multi-view graph autoencoder for graph clustering \cite{fan2020one2multi} named as One2Multi. One2Multi employs graph convolution neural network as the hidden layers of autocoder and uses the adjacency matrices of multiple views to extract the features of nodes. Then the clustering result of k-means algorithm is used to learn the parameters of graph convolution neural network by self-supervised learning.

Wang et al. propose an adaptive multi-channel graph attentional convolutional network \cite{wang2020gcn} called as AM-GCN. AM-GCN firstly computes the similarity matrix of nodes from the attribute features of nodes by cosine or kernel method and then constructs a feature graph by KNN method. The feature graph and the original graph are input into graph convolution neural network. AM-GCN considers that the feature graph and the original graph describe the same graph. Therefore there is common information between the structure features and the attribute features. AM-GCN designs three graph convolution neural networks to learn the low dimensional embedding vectors of the structure features, the attribute features and the common information, respectively. The final embedding vectors of nodes are the weighting fusion of the three groups of the embedding vectors.

Chen et al. propose a robust node representation learning algorithm named as CGNN \cite{chen2020learning}. CGNN introduces contrastive learning \cite{liu2020self} to train graph neural network in an unsupervised way. For each node, CGNN discards some neighborhood edges with a certain probability and it can obtain two local subgraphs of each node. Then the two low dimensional embedding vectors can be output after the two local subgraphs are input into graph neural network. The two low dimensional embedding vectors of the same node are seen as a positive sample pair since they describe the local structure of the same node. CGNN selects \emph{k} different low dimensional embedding vectors of other nodes as negative samples. CGNN introduces noise contrastive estimation to train graph neural network, therefore it has a good robustness for the representation learning of nodes.

Qiu et al. propose a graph contrastive coding graph neural network called as GCC \cite{qiu2020gcc}. GCC can be pre-trained on a large data set and then fine-tuned on a specific task. For each node, GCC samples one or multiple \emph{r}-ego networks. The \emph{r}-ego networks sampled from the same node are as the positive samples. The \emph{r}-ego networks sampled from the different nodes are as the negative samples. GCC considers that the positive sample pair should be as similar as possible and the negative samples should be far away from the positive samples. The parameters of the autocoder which is a graph neural network are trained by contrastive learning. The low dimensional embedding vectors can be obtained by inputting the \emph{r}-ego networks of nodes into the autocoder.

\section{The proposed algorithm}
For a graph $G=\left \langle V, \bm{A}, E, \bm{X} \right \rangle$, \emph{V} is the set of nodes in the graph, $\bm{A}\in \mathbb{R}^{n\times n}$ is the adjacency matrix of \emph{G}, \emph{E} is the edge set of \emph{G} and $\bm{X}$ is the attribute features of nodes. For $\forall v_{i},v_{j}\in V$, $\emph{A}_{ij}=1$ if $e_{ij}\in E$, otherwise, $\emph{A}_{ij}=0$. It is known that $\bm{A}$ is high dimensional and sparse. Graph embedding is to learn a map function \emph{f}: $V\mapsto \mathbb{R}^{n\times d}$ which projects each node into a \emph{d}-dimension space $\left ( d\ll n \right )$ and the similarity between nodes is preserved.
\subsection{Model overview}
It is common for the nodes of social networks without labels and the high-order information of graph can provide much semantic information for community discovery task. Therefore a self-supervised deep graph embedding algorithm with high-order information fusion is proposed for community discovery. Figure 1 is an illustration of the proposed algorithm (SDGE). SDGE firstly computes \emph{r} different high-order matrices of adjacency matrix. It is known that different high-order matrices present different semantic information. Then there are \emph{r} graph convolution neural networks operating on \emph{r} different high-order matrices, respectively. The outputs of the graph convolution neural networks are the \emph{r} features of nodes. After obtaining the initial embedding vectors from the \emph{r} graph convolution neural networks, the final embedding vectors are the fusion of the initial embedding vectors. With the final embedding vectors of nodes being obtained, the community partition of nodes can be assigned by multilayer perceptron or clustering algorithm. Considering the unlabeled nodes of graph, self-supervised learning approach is introduced to train the graph convolution neural networks. In the training process, data augmentation approach and negative sampling \cite{yang2020understanding} are employed to train the graph convolution neural networks effectively. In addition, SDGE can preserve the similarities of structure information and attribute information. The overview of SDGE is seen as Fig. \ref{fg1}.
\begin{figure*}
  \centering
  \includegraphics[width=19cm]{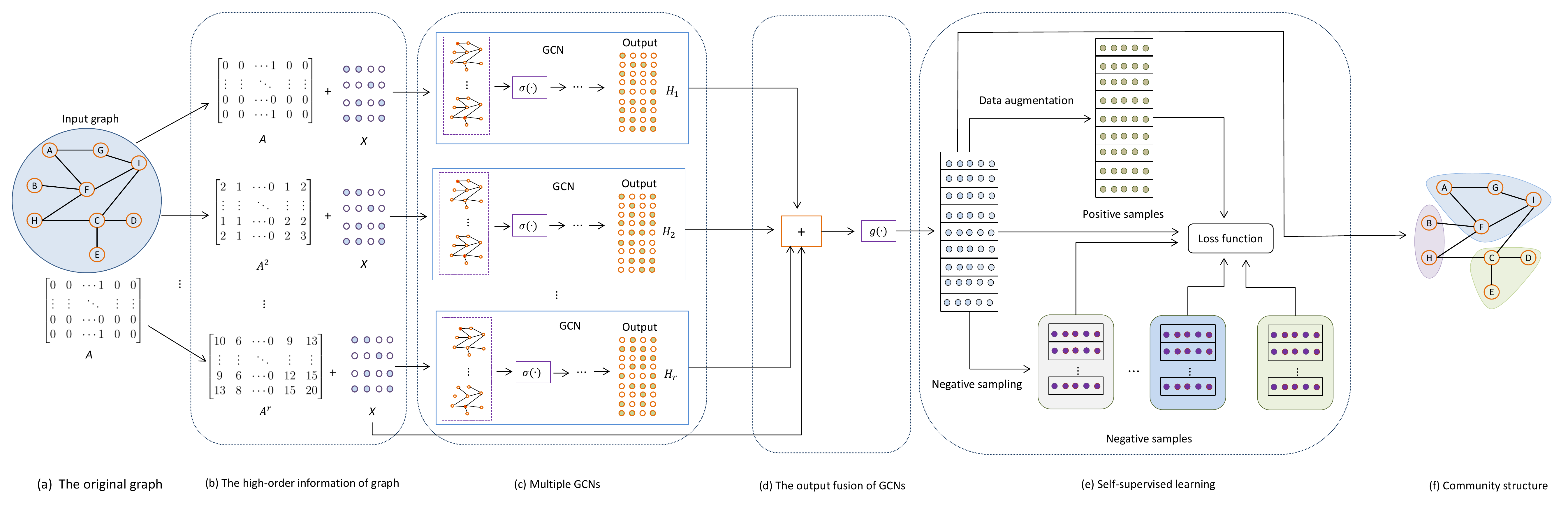}\\
  \caption{The overview of SDGE. It consists of six parts. (a): The original graph and the adjacency matrix are as the initial inputs. (b): Compute $\bm{A}^{r}$ with different \emph{r} values. The different multiplicative matrices represent random walks with different path lengths. (c): The different matrices of $\bm{A}^{r}$ are as the inputs of the \emph{r} GCNs, respectively. (d): The outputs of the \emph{r} GCNs are fused into the embedding vectors of nodes in original graph. (e): For each sample, new sample is generated by data augmentation which is seen as positive sample. Then SDGE samples several samples from the non-neighborhood nodes as the negative samples. The GCNs are trained by self-supervised learning approach. (f): The community structure is output by the end-to-end approach or clustering algorithm.}\label{fg1}
\end{figure*}
\subsection{Graph convolution neural networks with high-order information}
Let $\bm{A}$ be the adjacency matrix of the graph \emph{G} and $r\in \mathbb{Z}^{+}$ is the order of the adjacency matrix. $\bm{A}^{r}$ is defined as:
\begin{equation}\label{eq1}
  \bm{A}^{r} = \underset{r}{\underbrace{\bm{A}\cdot \bm{A}\cdot \cdot \cdot \bm{A}}}
\end{equation}
The value of \emph{r} determines the global and local information of the graph \emph{G}. A large \emph{r} means  $\bm{A}^{r}$ contains more global information.  A small \emph{r} means $\bm{A}^{r}$ contains more local information. In other words,  $\bm{A}^{r}$ is equal to the information of \emph{r}-step random walk in a graph. In order to make full use of local and global information, SDGE computes \emph{r} high-order multiplicative matrices of $\bm{A}$ which are as $\bm{A}, \bm{A}^{2}, \cdots, \bm{A}^{r}$. It can use $\bm{A}, \bm{A}^{2}, \cdots, \bm{A}^{r}$ as the inputs of \emph{r} graph convolution neural networks (GCNs), respectively. Therefore \emph{r} GCNs learn \emph{r} groups of the low dimensional embedding vectors of nodes from \emph{r} groups of the different high-order information of the graph \emph{G}.

Let $\bm{A}^{r}$ be the input of the \emph{r}th GCN $\left ( r=1,2,\cdots  \right )$ and $\bm{X}$ be the attribute features of nodes. The output of the \emph{r}th GCN in the $(\ell+1)$th layer is as follow:
\begin{equation}\label{eq2}
  \bm{H}_{r}^{\ell+1}=\delta \left (\widehat{\bm{D}}_{r}^{-\frac{1}{2}}\widehat{\bm{A}}_{r}\widehat{\bm{D}}_{r}^{-\frac{1}{2}}\bm{H}_{r}^{\ell}\bm{W}_{r}^{\ell}  \right )
\end{equation}
where $\bm{H}_{r}^{\ell}$ is the output of the \emph{r}th GCN in the $\ell$th layer, it is also the input of the $\left ( \ell+1 \right )$th layer and $\bm{H}_{r}^{0}=\bm{X}$; $\widehat{\bm{A}}_{r}=\bm{A}^{r}+\bm{I}$ and $\bm{I}$ is an identity matrix; $\bm{W}_{r}^{\ell}$ is the parameter of the $\ell$th layer and $\widehat{\bm{D}}_{r}$ is the degree matrix of $\widehat{\bm{A}}_{r}$; $\delta \left ( \cdot  \right )$ is the activation function. In general, different GCNs use different activation functions. For SDGE, Dynamic ReLU function is introduced and Chen et al. have proved that Dynamic ReLU can obtain better performance than ReLU \cite{chen2020dynamic}. In order to avoid over-fitting and improve the performance of SDGE, Batch Normalization \cite{ioffe2015batch} is introduced and the output of each hidden layer is normalized before it is input into the next hidden layer.

After inputting $\bm{A}, \bm{A}^{2}, \cdots, \bm{A}^{r}$ into \emph{r} GCNs, it can obtain \emph{r} groups of the outputs denoted as $\bm{H}_{1},\bm{H}_{2},\cdots,\bm{H}_{r}$ and $\bm{H}_{r}$ is the output of the \emph{r}th GCN in the last hidden layer. The final output $\bm{H}$ of GCNs is fused from the outputs of the \emph{r} GCNs by aggregation operator. The aggregation operator is defined as follow:
\begin{equation}\label{eq3}
  \bm{H}=Aggregate\left ( \bm{H}_{1},\bm{H}_{2},\cdots ,\bm{H}_{r} \right )
\end{equation}
where $Aggregate\left ( \cdot \right )$ is the aggregation function which can fuse the outputs of multiple GCNs. The most common aggregation functions include summation, mean, max and CONCAT, etc. In this paper, CONCAT or summation is selected as the aggregate function of SDGE because CONCAT and summation are injective \cite{xu2020powerful}. It is known that $\bm{H}_{1}$ contains more the local information of graph than $\bm{H}_{r}$ and $\bm{H}_{r}$ contains more the global information of graph than $\bm{H}_{1}$. The local and global information of graph play different roles on graph embedding. Therefore SDGE uses the weighting approach to fuse the outputs of multiple GCNs. For each output of GCN, it can use clustering algorithm to discover the community structure. After the community structure of $\bm{H}_{i}$ $\left ( i=1,2,\cdots ,r \right )$ is obtained, the modularity $Q_{i}$ of $\bm{H}_{i}$ which is a measure to reflect the strength of community structure \cite{newman2004finding} is computed as follow:
\begin{equation}\label{eq4}
  Q_{i} = \frac{1}{2\left | E \right |}\sum_{i^{'},j}^{\left | V \right |}\left ( A_{i^{'}j}-\frac{k_{i^{'}}\cdot k_{j}}{2\left | E \right |} \right )\cdot \sigma \left ( v_{i^{'}},v_{j} \right )
\end{equation}
where $\left | \cdot  \right |$ is the cardinality of a set, $k_{i^{'}}$ and $k_{j}$ are the degrees of the nodes $v_{i^{'}}$ and $v_{j}$, respectively. $\sigma \left ( v_{i^{'}},v_{j} \right )=1$ if $v_{i^{'}}$ and $v_{j}$ are in the same community; otherwise, $\sigma \left ( v_{i^{'}},v_{j} \right )=0$. The range of modularity is in $\left [ 0,1 \right ]$. A larger modularity means that the community structure is better which also reflects the output of GCN is better. Therefore the weight $\alpha_{i}$ of $\bm{H}_{i}$ is determined as follow:
\begin{equation}\label{eq5}
  \alpha_{i} = \frac{\exp\left ( Q_{i} \right )}{\sum\limits_{j=1}^{r}\exp\left ( Q_{j} \right )}
\end{equation}
If the aggregate function is summation or CONCAT, $\bm{H}$ is determined as follow:
\begin{equation}\label{eq6}
  \bm{H}=\sum\limits_{i=1}^{r}\alpha_{i}\bm{H}_{i} \ \ \ \text{or} \ \ \ \bm{H}=\left |  \right |_{i=1}^{r}\alpha_{i}\bm{H}_{i}
\end{equation}
where $\left |  \right |$ is the concatenate operator which concatenates the output matrices of GCNs. It is known that the attribute information of nodes is an important information for the partition of nodes. Therefore $\bm{X}$ is also concatenated to the output matrices of GCNs if the graph is an attribute graph. The final fusion is defined as follow:
\begin{equation}\label{eq61}
  \bm{H}\leftarrow \bm{H}\left |  \right | \bm{X}
\end{equation}

After the fusion of Eq.(\ref{eq61}) is obtained, the fusion is mapped nonlinearly by  multilayer perceptron (MLP). Chen et al.\cite{chen2020simple} prove that MLP is beneficial to self-supervised learning \cite{chen2020Asimple}. The final embedding result $\bm{Z}$ is determined as $\bm{Z}=g \left ( \bm{H} \right )\in \mathbb{R}^{n\times d}$ where $g \left (\cdot \right )$ is the nonlinear mapping function of MLP.
\subsection{The self-supervised learning of SDGE}
Self-supervised learning is an important unsupervised approach for graph neural network. It can train GCN from unlabeled graph data. The parameters of SDGE can be optimized by minimizing loss function. The diagram of self-supervised learning of SDGE is as Fig. \ref{fg2}.

For each node, SDGE employs data augmentation to generate positive samples. Let $\bm{Z}=\left [ \bm{z}_{1}^{T},\bm{z}_{2}^{T},\cdots ,\bm{z}_{n}^{T} \right ]^{T}$ and $\bm{z}_{i}\in\mathbb{R}^{1\times d}$ be the embedding vector of the node $v_{i}$.  For $\forall v_{i}\in V$, the noise, such as Gaussian noise, etc., is added into $\bm{z}_{i}$ and the new generated data is denoted as $\bm{z}_{i}^{+}\in \mathbb{R}^{1\times d}$. $\bm{z}_{i}^{+}$ and $\bm{z}_{i}$ can be seen as a positive sample pair since both of them are the descriptions of the node $v_{i}$. For the node $v_{i}$, it can also sample \emph{m} samples from the non-neighborhood nodes of $v_{i}$ which are seen as negative samples. The similarity between $\bm{z}_{i}$ and $\bm{z}_{i}^{+}$ should be as large as possible. The similarity between $\bm{z}_{i}$ and the negative samples of $\bm{z}_{i}$ should be as small as possible. Therefore the loss function of the self-supervised learning is defined as follow:
\begin{equation}\label{eq7}
  \mathcal{L}_{s}=\frac{1}{n}\sum_{i=1}^{n}\log\delta \left ( \bm{z}_{i}\cdot \bm{z}_{i}^{+}/\tau \right )-E_{\bm{z}_{j}\sim P_{n}\left ( v \right ), v_{j}\notin N\left ( v_{i} \right )}\left [ \log\delta \left ( \bm{z}_{i}\cdot \bm{z}_{j}/\tau \right ) \right ]
\end{equation}
In Eq.(\ref{eq7}), the first term is the similarity of the positive sample pair; the second term is the similarity between the node $v_{i}$ and the negative samples where $P_{n}\left ( v \right )$ is the distribution of negative samples, $N\left ( v_{i} \right )$ is the neighbors of $v_{i}$ and $\tau$ is the temperature.
\begin{figure}
  \centering
  \includegraphics[width=8.6cm]{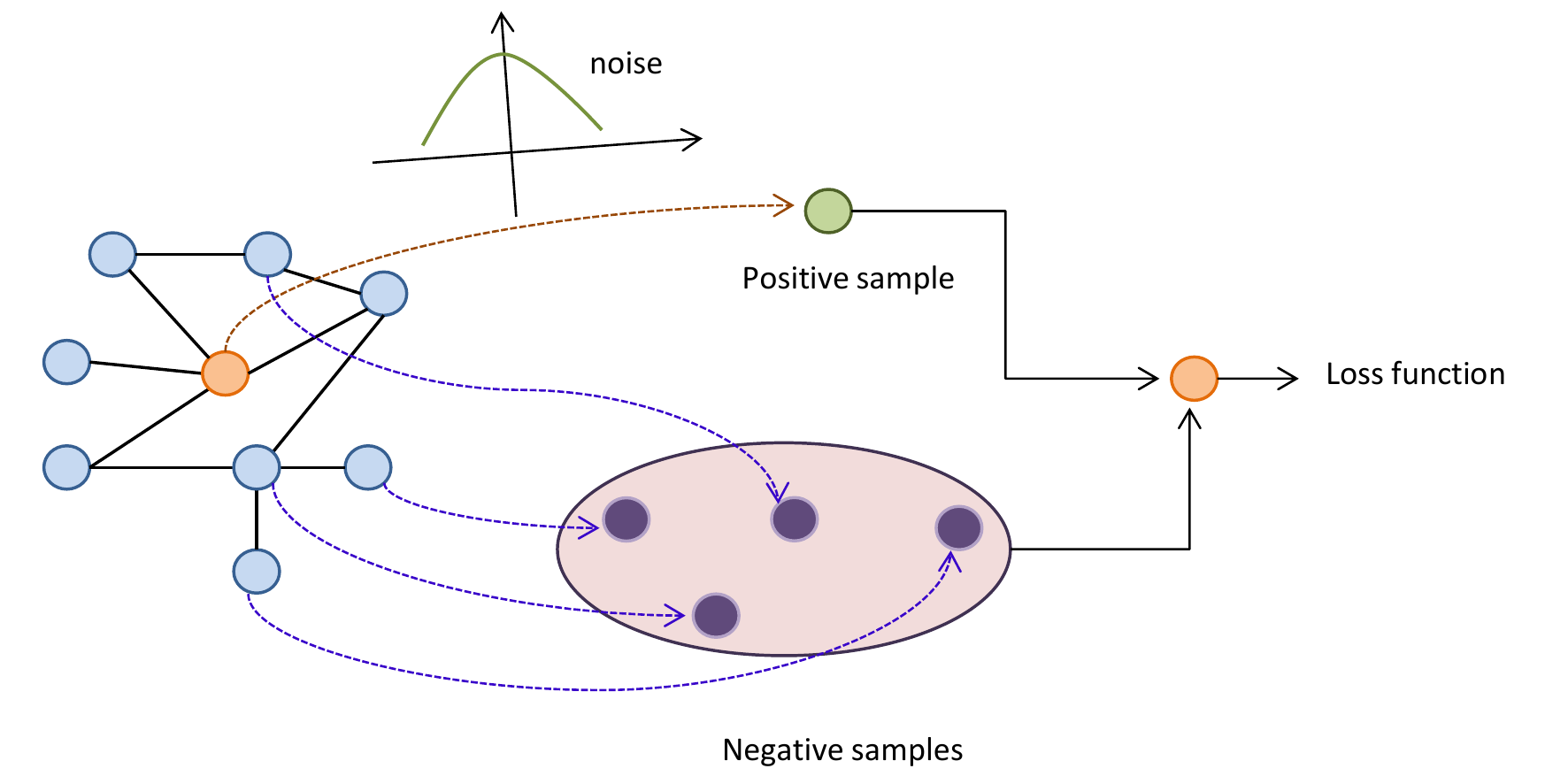}\\
  \caption{The self-supervised learning of SDGE. For each node, the positive sample is generated by adding noise into the original sample and the negative samples are sampled from the non-neighborhood nodes.}\label{fg2}
\end{figure}
It is known from Eq.(\ref{eq2}) that the output of MLP includes the structure information and the attribute information of nodes. Therefore the low dimensional embedding $\bm{Z}$ should preserve the similarities of the structure information and the attribute information after dimension reduction. The difference between $\bm{Z}$ and the structure information should be as small as possible and the difference between $\bm{Z}$ and the attribute information should be also as small as possible. The loss of the structure information and the attribute information of nodes are defined as follow:
\begin{equation}\label{eq8}
  \mathcal{L}_{sa}=\frac{1}{n^{2}}\left \| \bm{Z}\bm{Z}^{T}-\bm{A}\right \|_{F}^{2}+\frac{1}{n^{2}}\left \| \bm{Z}\bm{Z}^{T}-\bm{X}\bm{X}^{T}\right \|_{F}^{2}
\end{equation}
In Eq.(\ref{eq8}), the first term is the loss of the structure information and the second term is the loss of the attribute information. SDGE is to preserve the minimal loss of the structure information and the attribute information.

SDGE is to project the nodes of graph into the low dimensional space. Therefore the embedding result of SDGE should also satisfy the constraint of graph regularization. It means the similar nodes in high dimensional space are also similar in the dimensional embedding space. The loss function $\mathcal{L}_{r}$ is defined as:
\begin{equation}\label{eq9}
  \mathcal{L}_{r}= \frac{1}{n}\bm{tr}\left ( \bm{Z}^{T}\bm{L}\bm{Z} \right )
\end{equation}
where $\bm{L}$ is the Laplace matrix of graph \emph{G}. $\bm{L}=\bm{W}-\bm{D}$, $\bm{W}$ is the weight matrix of the edges in the graph \emph{G}, $\bm{D}$ is the degree matrix and $D_{ii}=\sum_{j=1}^{n}W_{ij}$.

From Eqs.(\ref{eq7})-(\ref{eq9}), the optimization problem of the loss function $\mathcal{L}$ is defined as:
\begin{equation}\label{eq10}
  \underset{\bm{Z}}{\min}\ \mathcal{L}=\beta\cdot\mathcal{L}_{sa}+\gamma\cdot\mathcal{L}_{r}-\mathcal{L}_{s}
\end{equation}
where $\beta\geq 0$ and $\gamma\geq 0$ are the predefined parameters. The parameters of GCNs and MLP are adjusted iteratively by backward propagating the loss of Eq.(\ref{eq10}).

$\bm{Z}$ can be solved from Eq.(\ref{eq10}). Then spectral propagation \cite{zhang2019prone} is introduced to enhance the embedding result. $\bm{Z}$ is updated as:
\begin{equation}\label{eq11}
  \bm{Z}\leftarrow \bm{D}^{-1} \bm{A}\left ( \bm{I}-\widetilde{\bm{L}} \right )\bm{Z}
\end{equation}
where $\widetilde{\bm{L}}=\bm{U}g\left ( \bm{\Lambda } \right )\bm{U}^{-1}$ is the modulated Laplacian and \emph{g} is as the spectral modulator. $\bm{U}$ and $\bm{\Lambda }$ are the eigenvalue decomposition results of the random walk normalized graph Laplacian $\bar{\bm{L}}=\bm{I}-\bm{D}^{-1}\bm{A}$ where $\bar{\bm{L}}=\bm{U}\bm{\Lambda }\bm{U}^{-1}$. $\bm{\Lambda }=diag\left ( \left [ \lambda_{1},\lambda_{2},\cdots ,\lambda_{n} \right ] \right )$. $\lambda_{i}$ $\left ( i=1,2,\cdots ,n \right )$ is the eigenvalue of $\bar{\bm{L}}$ and $0=\lambda_{1}\leq\lambda_{2}\leq\cdots \leq\lambda_{n}$. The reference \cite{zhang2019prone} presents truncated Chebyshev expansion to compute the modulated Laplacian efficiently.

If SDGE is an end-to-end approach, the dimension \emph{d} of $\bm{Z}$ which is the output of MLP is set to the number of communities and the community partition of the input graph \emph{G} is determined by $softmax(\cdot)$. If SDGE is not an end-to-end approach, the community structure of the graph \emph{G} can be obtained by the clustering algorithm after MLP outputs the low dimensional embedding vectors $\bm{Z}$.

Therefore the detail steps of SDGE algorithm are summarized as \textbf{Algorithm 1}.
\begin{algorithm}
        \caption{SDGE}
        \begin{algorithmic}[1] 
            \Require A graph $G=\left \langle V,\bm{A},E,\bm{X}\right \rangle$, the number of clusters \emph{k} and the dimension after dimension reduction \emph{d}; the parameters $\beta$ and $\gamma$ and \emph{r}.
            \Ensure The community structure $\bm{C}\in \mathbb{R}^{n\times k}$.
            \State Compute $\bm{A}, \bm{A}^{2}, \cdots, \bm{A}^{r}$;
            \State Initialize \emph{r} GCNs and MLP;
            \State Input $\bm{A}, \bm{A}^{2}, \cdots, \bm{A}^{r}$ into \emph{r} GCNs, respectively;
            \While{The loss $\mathcal{L}$ is not convergent}
            \State Obtain the embedding vectors $\bm{Z}$;
            \State Compute the loss $\mathcal{L}$ from MLP;
            \State Back propagate the loss $\mathcal{L}$ and adjust the parameters of GCNs and MLP;
            \State Enhance the embedding vectors $\bm{Z}$ by spectral propagation;
            \EndWhile
            \State Obtain the community partition $\bm{C}$ of graph \emph{G} from MLP or clustering algorithm.
        \end{algorithmic}
\end{algorithm}\label{ag1}

In \textbf{Algorithm 1}, SDGE can effectively fuse the structure information and the attribute information. In addition, multiple GCNs are utilized to extract the features of nodes which is equivalent to multi-view learning. Step 1 is to compute the high-order information of graph and the high-order information matrices include the local and global information. Steps 2-3 input the high-order information matrices into different GCNs. Therefore the embedding result integrates the local and global information of graph and the attribute information of nodes. Steps 4-8 are to train GCNs and MLP. Step 6 is to enhance the embedding result by spectral propagation. Step 10 obtains the community partition by the end-to-end approach or clustering algorithm.

\subsection{The complexity analysis}
For SDGE, the time complexity of Eq.(\ref{eq1}) is $\mathcal{O}\left ( rn^{3} \right )$. Let \emph{t} be the number of iterations of GCN and MLP and \emph{l} be the number of layers. The time complexity of GCNs is $\mathcal{O}\left (rtln^{3}\right )$. The time complexity of MLP is $\mathcal{O}\left (tnd^{2}\right )$. It costs $\mathcal{O}\left (rn^{2}\right )$ to compute $\bm{\alpha}$. The time complexity of the loss $\mathcal{L}$ is $\mathcal{O}\left (tnd^{2}+tn^{2}d+tnmd^{2}\right )$. It is known that $rtln^{3}\geq tnd^{2}$ and $tnd^{2} \leq tn^{2}d \leq rtln^{3}$. Therefore the time complexity of SDGE is $\mathcal{O}\left ( rtln^{3}+ tnmd^{2}\right )$.
\subsection{The performance analysis}
In Eq.(\ref{eq6}), the outputs of the GNCs are fused by weighting mechanism. It can use the high-order information of graph to project the representation of nodes into the low dimensional space effectively.
\begin{theorem}\label{th1}
In Eq.(\ref{eq6}), the fusion with weighting mechanism is of benefit to the output of MLP if the aggregation function is $sum\left ( \cdot  \right )$.
\end{theorem}

\begin{proof}
Let $f_{1},f_{2},\cdots,f_{r}$ be the outputs of the GCNs and $\alpha_{1},\alpha_{2},\cdots,\alpha_{r}$ be the weights of the outputs $f_{1},f_{2},\cdots,f_{r}$.\\
$\therefore$ The fusion $f=\alpha_{1}f_{1}+\alpha_{2}f_{1}+\cdots+\alpha_{r}f_{r}=\sum_{i=1}^{r}\alpha _{i}f_{i}$.\\
$\because$ $\alpha _{1}+\alpha _{2}+\cdots +\alpha _{r}=1$ and $0\leq \alpha_{i}\leq 1$ $\left ( i=1,2,\cdots,r \right )$.\\
$\therefore$ $\mathbb{E}\left [ f \right ]=\sum_{i=1}^{r}\alpha _{i}f_{i}$ is the expectation of $f_{1},f_{2},\cdots,f_{r}$.\\
Let $f^{\ast}$ be the real target of MLP and $F\left ( \cdot  \right )$ be the map function of MLP.\\
From Jensen's inequality \cite{Zhou2020introduction}, it concludes as follow:
\begin{equation}\label{eq12}
\begin{aligned}
  &\left \| F\left ( \mathbb{E}\left [ \sum_{i=1}^{r}\alpha _{i}f_{i} \right ] \right )-F\left ( f^{\ast } \right ) \right\| \\
  &=\left \| F\left ( \mathbb{E}\left [ f \right ] \right )-F\left ( f^{\ast } \right ) \right \| \leq \left \| \mathbb{E}\left [ F\left ( f \right )\right ]-F\left ( f^{\ast } \right ) \right \|
\end{aligned}
\end{equation}
Therefore it is known that the output error can be decreased by the weighting mechanism.
\end{proof}

\begin{theorem}\label{th2}
In Eq.(\ref{eq6}), the fitting capability of MLP with the aggregation function $CONCAT(\cdot)$ is not weaker than the aggregation function $sum\left ( \cdot  \right )$.
\end{theorem}
\begin{proof}
Let $f_{1},f_{2},\cdots,f_{r}$ be the outputs of the GCNs.\\
$\therefore$ $\sum_{i=1}^{r}f_{i}\in \mathbb{R}^{n\times l}$, $\left [ f_{1},f_{2},\cdots,f_{r}\right ]\in \mathbb{R}^{n\times rl}$ and $l<lr$.\\
Let $l^{'}$ be the hidden nodes' number of MLP.\\
$\therefore$ The parameters of MLP with the aggregation function $sum\left ( \cdot  \right )$ is $l\cdot l^{'}$ and the parameters of MLP with the aggregation function $CONCAT(\cdot)$ is $rl\cdot l^{'}$ $\left ( r\geq 1 \right )$.\\
$\therefore$ The parameters of MLP with the aggregation function $sum\left ( \cdot  \right )$ is not more than the MLP with the aggregation function $CONCAT(\cdot)$.\\
$\therefore$ The fitting capability of MLP with the aggregation function $CONCAT(\cdot)$ is not weaker than the aggregation function $sum\left ( \cdot  \right )$.
\end{proof}

In Theorem \ref{th2}, the fitting capability of MLP will be too strong and over-fitting will also appear if $r$ is a large value. In addition, the time cost will also increase with the increase of \emph{r} value.

\section{The experimental results and analysis}
The proposed algorithm and the comparison algorithms are conducted on a server with Ubuntu 18.04.2 operating system, Intel Core i9-7900X CPU and 64G RAM. The proposed algorithm is implemented by Python 3.7 and the executable codes of the comparison algorithms are from the released codes in the papers.
\subsection{The comparison algorithms and the parameters of the proposed algorithm}
\subsubsection{The comparison algorithms}
In order to verify the performance of the proposed algorithm, the follow algorithms are selected as the comparison algorithms of this paper:
\begin{itemize}
  \item \textbf{Line} \cite{tang2015line}: it is a graph embedding algorithm based on the assumption of neighborhood similarity and introduces first-order proximity and second-order proximity to define the similarity between vertices in graph.
  \item \textbf{struc2vec} \cite{ribeiro2017struc2vec}: it is a graph embedding algorithm based on random walk and introduces structure similarity to define the similarity of any two nodes.
  \item \textbf{DANMF} \cite{ye2018deep}: it stacks multiple non-negative matrix factorization approaches as the autocoder and decoder layers to learn the final community assignment.
  \item \textbf{Graph2gauss} \cite{bojchevski2018deep}: it can efficiently learn versatile node embedding on large scale graph and embed nodes as Gaussian distribution to capture the uncertainty of the representation.
  \item \textbf{Modsoft} \cite{hollocou2019modularity}: it can discover the community structure by maximizing the modularity of the node partition and uses sparse matrix to record the partition result to improve the efficiency.
  \item \textbf{ProNE} \cite{zhang2019prone}: it is a fast and scalable network representation learning algorithm and the node embedding vectors can be efficiently obtained by the randomized tSVD approach.
  \item \textbf{DGI} \cite{velickovic2019deep}: it is a self-supervised graph neural network to learn structure information of a graph and aims at maximizing mutual information between the input data and output data.
  \item \textbf{GIC} \cite{mavromatis2020graph}: it is a self-supervised graph neural network and uses cluster-level node information to learn the dimensional embedding vectors of nodes.
  \item \textbf{GMI} \cite{peng2020graph}: it is a self-supervised graph neural network to learn the embedding vectors of nodes by mutual Information maximization.
  \item \textbf{SDGE-cat}: it is the SDGE algorithm with $CONCAT\left ( \cdot  \right )$ to aggregate the outputs of GCNs.
  \item \textbf{SDGE-sum}: it is the SDGE algorithm with $sum\left ( \cdot  \right )$ to aggregate the outputs of GCNs.
\end{itemize}

\subsubsection{The parameters of the proposed algorithm} For SDGE, \emph{k} is set to the real community number of data set. $\beta,\gamma\in \left [ 0,1 \right ]$, $\tau\in \left [ 0,100 \right ]$, $d=\left \{ 64,128 \right \}$ and $r=4$. For the GCNs in SDGE structure, the GCNs are with four layers structure and the neuron number of the different layers is $\left [ 200,170,140,100 \right ]$. The activation function of MLP is \emph{sigmoid}. The aggregation function is CONCAT or \emph{sum}.
\subsection{The experimental data sets}
In order to test the performance of SDGE and the comparison algorithms, the following data sets are selected as the experimental data sets:
\begin{itemize}
  \item \textbf{ACM} is an author relationship network from ACM Digital Library. The nodes represent papers and there is an edge between two nodes if there is the same author in two papers.
  \item \textbf{USA} is an air-traffic network of USA. A node represents an airport and an edge means the existence of commercial flight between two airports.
  \item \textbf{Image} is an image segmentation data set. The images are hand-segmented and each instance is a $3\times 3$ region. The attribute graph is constructed by KNN $\left ( K=10 \right )$.
  \item \textbf{Hyperplane} is an artificial data set from MOA platform\footnote{\url{https://moa.cms.waikato.ac.nz/downloads/}}. For a data point $\bm{x}=\left [ x_{1},x_{2},\cdots ,x_{d} \right ]\in \mathbb{R}^{d}$ and a \emph{d}-dimension hyperplane $\sum_{i=1}^{d}a_{i}x_{i}=a_{0}$, if $\sum_{i=1}^{d}a_{i}x_{i}\geq a_{0}$, $\bm{x}$ is marked as a positive sample, otherwise, $\bm{x}$ is marked as a negative sample. The attribute graph is constructed by KNN $\left ( K=10 \right )$.
  \item \textbf{Waveform} is a wave data set. Each sample is generated from a combination of 2 of 3 base waves. The attribute graph is constructed by KNN $\left ( K=10 \right )$.
\end{itemize}
The detailed information of the experimental data sets are presented in Table \ref{tb1}.
\begin{table}
\caption{The details of the experimental data sets.}
 \begin{tabular}{ccccccccccc}\toprule
  Data sets & \#nodes &\#edges &\#attribute &\#communities\\\midrule
  ACM & 3,025 & 26,256 &\textcolor[rgb]{0.85,0.00,0.00}{\XSolidBrush} &3 \\
  USA & 1,190 & 13,599 &\textcolor[rgb]{0.85,0.00,0.00}{\XSolidBrush} &4 \\
  Image & 2,100 &21,000 &\textcolor[rgb]{0.85,0.00,0.00}{\CheckmarkBold}$\left (19\right )$ &7 \\
  Hyperplane &4,000 &40,000 &\textcolor[rgb]{0.85,0.00,0.00}{\CheckmarkBold}$\left (40\right )$ &2 \\
  Waveform &3,500 &35,000 &\textcolor[rgb]{0.85,0.00,0.00}{\CheckmarkBold}$\left (21\right )$ &3 \\\bottomrule
 \end{tabular}\label{tb1}
\end{table}
\subsection{Evaluation criteria}
In order to evaluate the performance of the algorithms, the following evaluation criteria are used:\\
(1) Jaccard index (J):
\begin{equation}\label{eq12}
  J = \frac{TP}{TP+FN+FP}
\end{equation}
(2) Folkes and Mallows index (FM):
\begin{equation}\label{eq13}
  FM = \frac{TP}{\sqrt{\left ( TP+FN \right )\cdot \left ( TP+FP \right )}}
\end{equation}
(3) F1-measure ($F_{1}$):
\begin{equation}\label{eq14}
  F_{1}=\frac{2\cdot prcision\cdot recall}{prcision+recall}
\end{equation}
(4) Kulczynski index (K):
\begin{equation}\label{eq15}
  K=\frac{1}{2}\left ( \frac{TP}{TP+FP}+\frac{TP}{TP+FN} \right )
\end{equation}
where $prcision=TP/\left ( TP+FP \right )$ and $recall = TP/\left ( TP+FN \right )$. \emph{TP}, \emph{FP}, \emph{TN} and \emph{FN} are from confusion matrix. \emph{TP} is the number of data point pairs that they are in the same cluster and the real labels of the two data points are also the same. \emph{TN} is the number of data point pairs that they are in different clusters and the real labels of the two data points are also different. \emph{FP} is the number of data point pairs that they are in the same cluster and the real labels of the two data points are different. \emph{FN} is the number of data point pairs that they are in different clusters and the real labels of the two data points are the same.\\

\subsection{The performance on community discovery}
\subsubsection{The test results of the proposed algorithm and the comparison algorithms}
In order to test the effectiveness of the proposed algorithm, SDGE and the comparison algorithms are conducted on the five experimental data sets and the test results are showed as Tables \ref{tb2}-\ref{tb6}.
\begin{table}[h]
\caption{The experimental results of the algorithms on ACM data set.}
\centering
\begin{tabular}{cccccc}\toprule
Algorithms & J & FM & $F_{1}$ & K \\\midrule
Line &0.2514 &0.4094 &0.4018 &0.4171 \\
struc2vec &0.2568  &0.4216  &0.4079  &0.4361  \\
DANMF &0.1160  &0.2296  &0.2079  &0.2535 \\
Graph2gauss &0.2980  &0.4596  &0.4592  &0.4600 \\
Modsoft &0.0131  &0.1000  &0.0259  &0.3867 \\
ProNE &0.2752  &0.4510  &0.4316  &0.4714 \\
DGI &0.2268 &0.3701 &0.3697 &0.3705 \\
GIC &0.2998 &0.4880 &0.4611 &0.5165 \\
GMI &0.2485 &0.3982 &0.3979 &0.3986 \\
SDGE-cat &\textbf{0.3201}  &\textbf{0.5442}  &\textbf{0.4849}  &\textbf{0.6107} \\
SDGE-sum &0.3158  &0.5335  &0.4801  &0.5929 \\\toprule
\end{tabular}\label{tb2}
\end{table}

\begin{table}[h]
\caption{The experimental results of the algorithms on USA data set.}
\centering
\begin{tabular}{cccccc}\toprule
Algorithms & J & FM & $F_{1}$ & K \\\midrule
Line &0.2672 &0.4239 &0.4213 &0.4266 \\
struc2vec &0.1751  &0.3060  &0.2977  &0.3148 \\
DANMF &0.0218  &0.0762  &0.0427  &0.1359 \\
Graph2gauss &0.1718  &0.2979  &0.2931  & 0.3027\\
Modsoft &0.1057  &0.1967  &0.1911  &0.2024 \\
ProNE &0.1813  &0.3173  &0.3067  &0.3284 \\
DGI &\textbf{0.4859} &\textbf{0.6512} &\textbf{0.6511} &\textbf{0.6513}\\
GIC &0.3334 &0.5001 &0.5000 &0.5002 \\
GMI &0.2021 &0.3363 &0.3363 &0.3363 \\
SDGE-cat &0.1799  &0.3128  &0.3049  &0.3208 \\
SDGE-sum &0.2022  &0.3590  &0.3364  &0.3832 \\\toprule
\end{tabular}\label{tb3}
\end{table}

\begin{table}[h]
\caption{The experimental results of the algorithms on Image data set.}
\centering
\begin{tabular}{cccccc}\toprule
Algorithms & J & FM & $F_{1}$ & K \\\midrule
Line &0.1396 &0.2545 &0.2448 &0.2646 \\
struc2vec &0.0863 &0.1594  &0.1594  &0.1600 \\
DANMF &0.0884  &0.0884  &0.1625  &0.4106 \\
Graph2gauss &0.2204  &0.3614  &0.3613  &0.3614 \\
Modsoft &0.0649  &0.2273  &0.1219  &0.4238 \\
ProNE &0.1574  &0.2733  &0.2715  &0.2751 \\
DGI &0.2582 &0.4119 &0.4104 &0.4134\\
GIC &0.2446 &0.3939 & 0.3931 &0.3947 \\
GMI &0.2375 &0.3847 &0.3838 &0.3856 \\
SDGE-cat &0.1641  &0.3749  &0.2820  &\textbf{0.4985} \\
SDGE-sum &\textbf{0.2759}  &\textbf{0.4364}  &\textbf{0.4325}  &0.4403 \\\toprule
\end{tabular}\label{tb4}
\end{table}

\begin{table}[h]
\caption{The experimental results of the algorithms on Hyperplane data set.}
\centering
\begin{tabular}{cccccc}\toprule
Algorithms & J & FM & $F_{1}$ & K \\\midrule
Line &0.4650 &0.6592 &0.6348 &0.6846 \\
struc2vec &0.3701  &0.5427  &0.5394  &0.5460  \\
DANMF &0.0184  &0.1012  &0.1012  &0.2839 \\
Graph2gauss &0.3335  &0.5002  &0.5002  &0.5002 \\
Modsoft &0.0314  &0.1306  &0.0609  &0.2797 \\
ProNE &0.3369  &0.5040  &0.5040  &0.5040 \\
DGI &0.4108 &0.5904 &0.5823 &0.5986 \\
GIC &0.3352 &0.5021 &0.5021 &0.5021 \\
GMI &0.4190 &0.6004 &0.5905 &0.6105 \\
SDGE-cat &\textbf{0.4968}  &\textbf{0.7027}  &\textbf{0.6639}  &\textbf{0.7438} \\
SDGE-sum &0.4963  &0.7019  &0.6634  &0.7426 \\\toprule
\end{tabular}\label{tb5}
\end{table}

\begin{table}[h]
\caption{The experimental results of the algorithms on Waveform data set.}
\centering
\begin{tabular}{cccccc}\toprule
Algorithms & J & FM & $F_{1}$ & K \\\midrule
Line &0.3232 &0.4903 &0.4883 &0.4924\\
struc2vec &0.2554  &0.4183  &0.4066  &0.4303  \\
DANMF &0.0372  &0.1692  &0.0718  &0.3990 \\
Graph2gauss &0.3521  &0.5209  &0.5208  &0.5209 \\
Modsoft & 0.2603 &0.4532  &0.4130  &0.4973 \\
ProNE &0.3466  &0.5148  &0.5148  &0.5149 \\
DGI &0.3373 &0.5045 &0.5045 &0.5045 \\
GIC &0.3413 &0.5089 &0.5089 &0.5089 \\
GMI &0.3391 &0.5065 &0.5065 &0.5065 \\
SDGE-cat &\textbf{0.4386}  &\textbf{0.6223}  &\textbf{0.6098}  &\textbf{0.6351} \\
SDGE-sum &0.3892  &0.5723  &0.5603  &0.5846 \\\toprule
\end{tabular}\label{tb6}
\end{table}
Tables \ref{tb2}-\ref{tb6} present the results of the proposed and comparison algorithms on the experimental data sets. From the results, it can be seen that SDGE-cat obtains the best results on ACM,  Hyperplane and Waveform data sets. SDGE-sum obtains the best results on Image data set except for Kulczynski index. SDGE-cat obtains the best result of Kulczynski index on Image data set. Therefore SDGE outperformances the comparison algorithms on the most data sets except for USA data set. On USA  data set, SDGE-cat is the fourth best of the eleven algorithms and DGI is the best of all. For USA  data set, there are only 1,190 nodes and the scale of the data set is not too large. However, there are \emph{r} GCNs and a MLP in SDGE. Each GCN is with four layers and MLP is with double layers. By contrast, the structure of neural network is too complex for USA data set. It means that the number of the parameters of neural network is not consistent with the complexity of data set. The too complex neural network reduces the performance of SDGE. Therefore the neural network algorithms with the shallower layers such as DGI, GIC, etc., outperform SDGE on USA data set.
\subsubsection{The effects of the parameters}
In order to test the effects of the parameters on the performance of the proposed algorithm, SDGE is conducted on ACM data set and the parameters $\tau$, $\beta$ and $\gamma$ are set to the different values. The test results are showed as Figs. \ref{fg3}-\ref{fg5}.
\begin{figure*}
  \begin{minipage}{9cm}
    \centerline{\includegraphics[width=7cm]{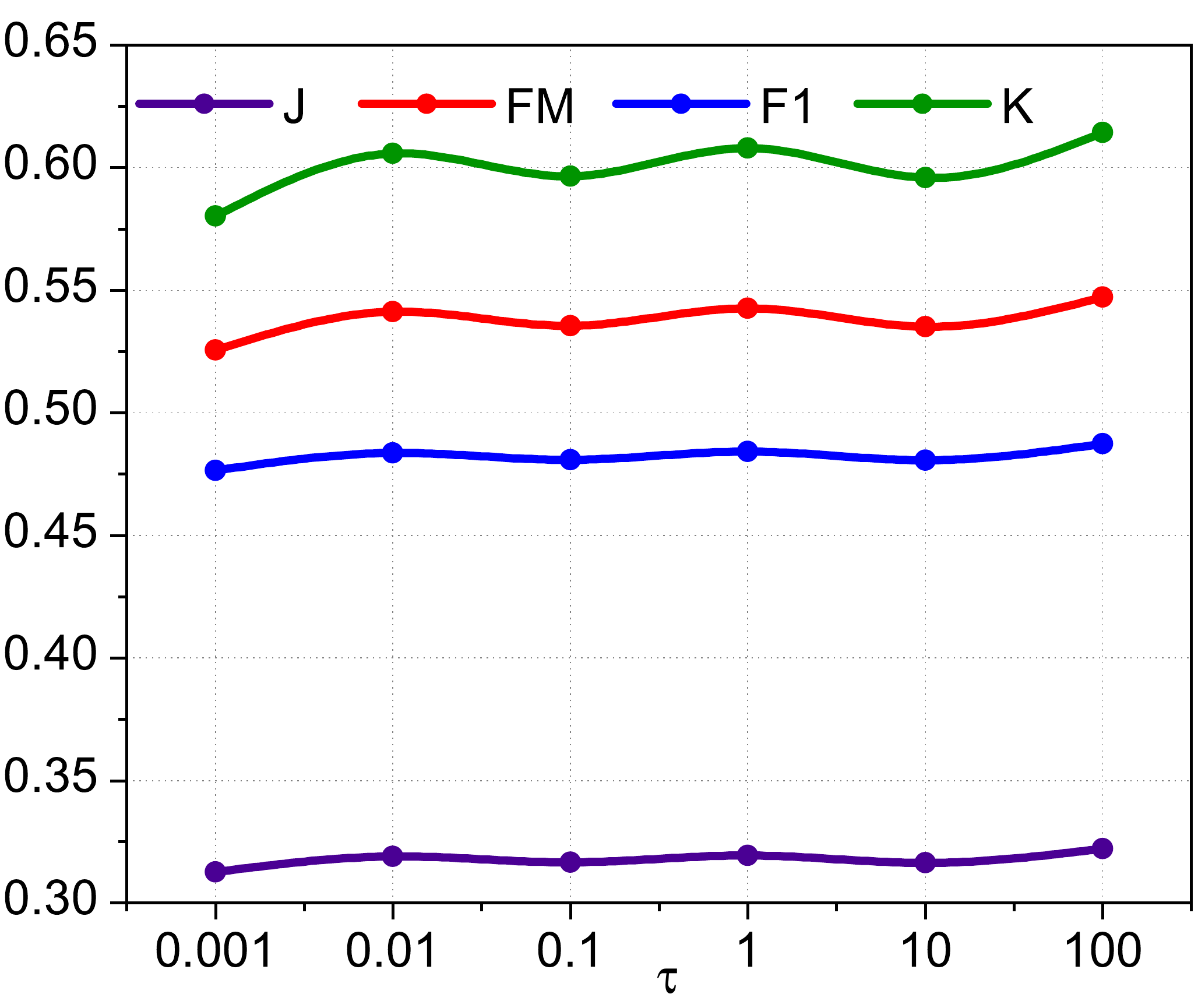}}
    \centerline{(a) SDGE-cat}
  \end{minipage}
  \hfill
  \begin{minipage}{9cm}
    \centerline{\includegraphics[width=7cm]{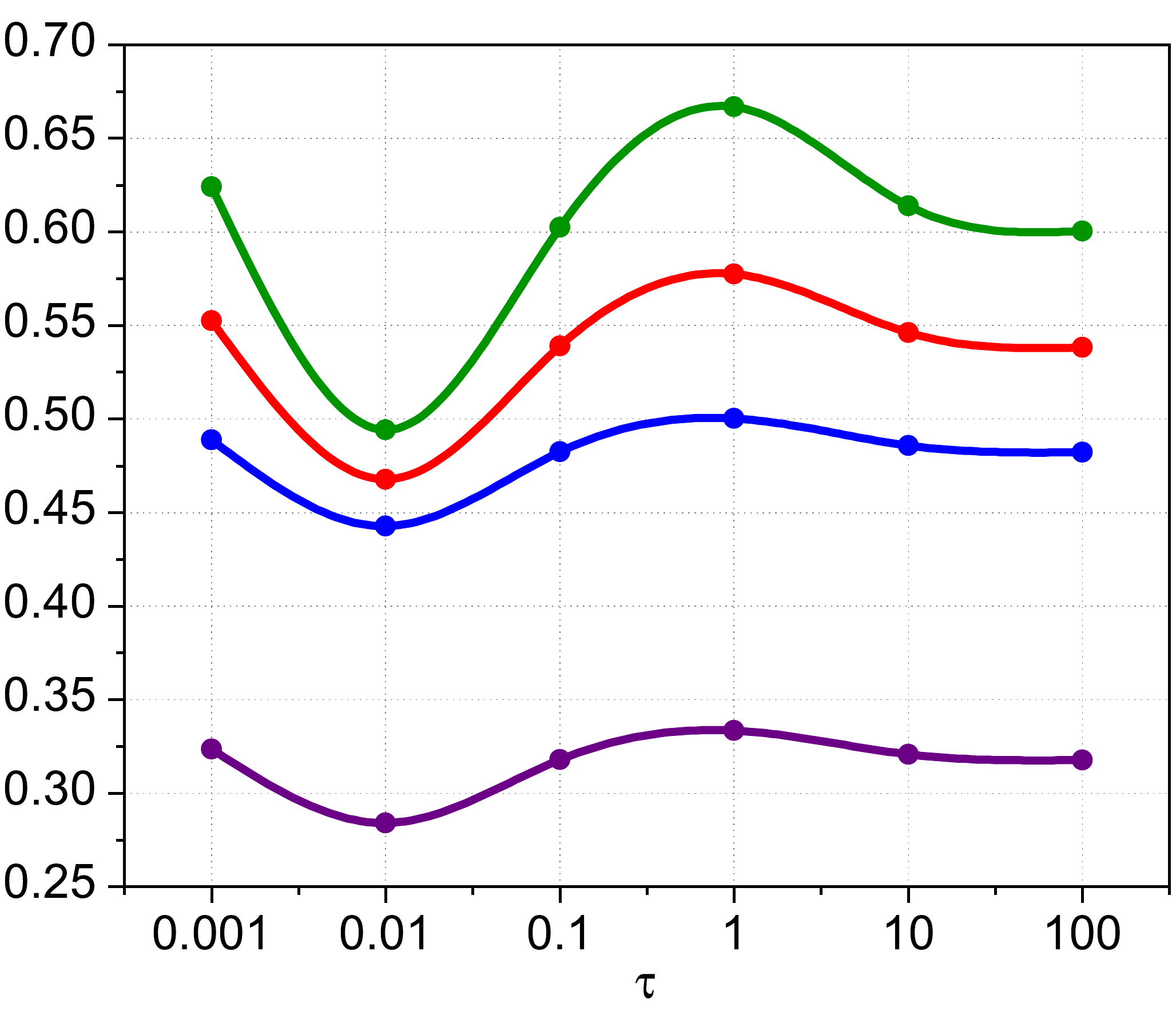}}
    \centerline{(b) SDGE-sum}
  \end{minipage}
\centering
\caption{The experimental results of SDGE with different $\tau$ values.}\label{fg3}
\end{figure*}

\begin{figure*}
  \begin{minipage}{9cm}
    \centerline{\includegraphics[width=7cm]{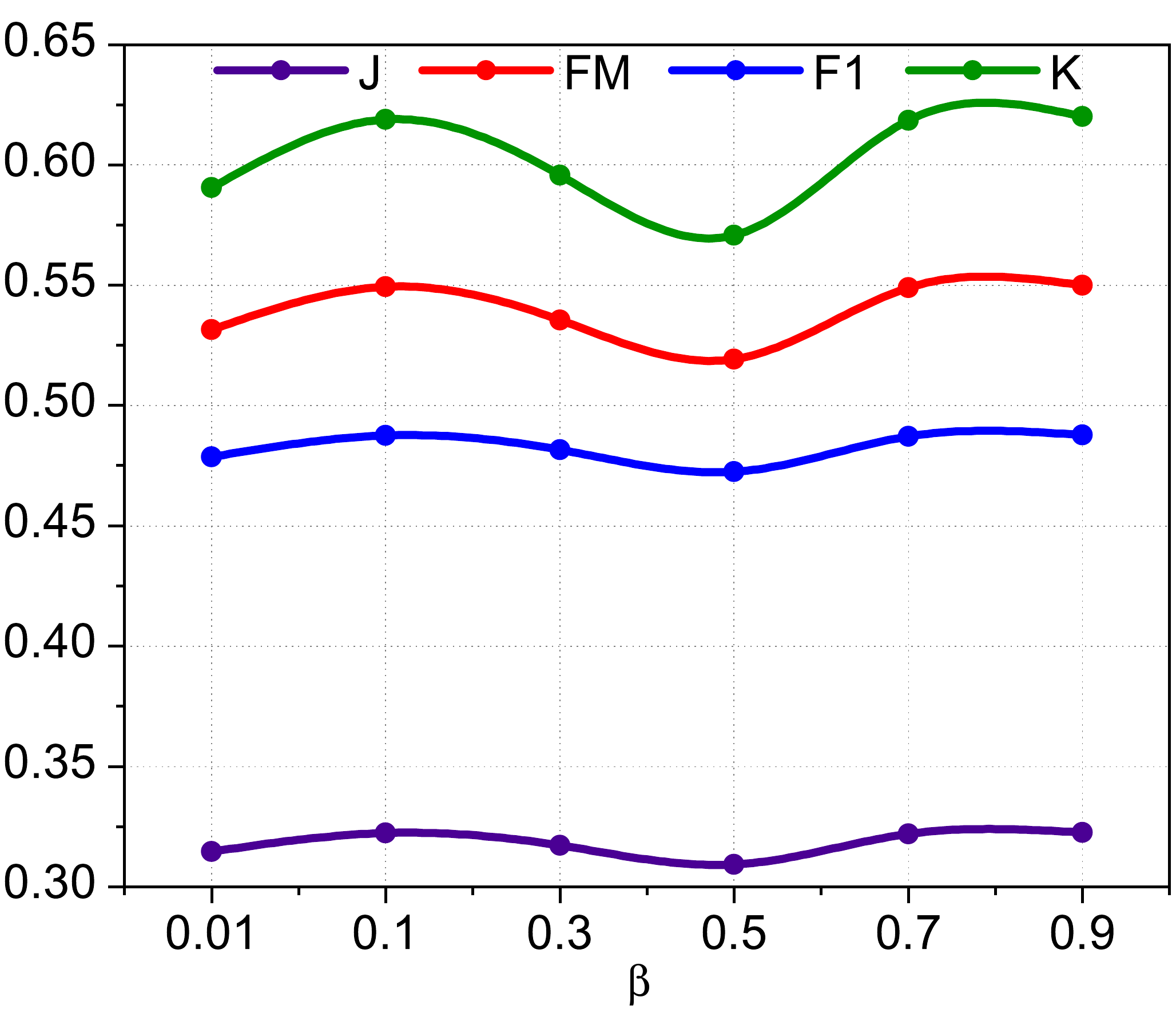}}
    \centerline{(a) SDGE-cat}
  \end{minipage}
  \hfill
  \begin{minipage}{9cm}
    \centerline{\includegraphics[width=7cm]{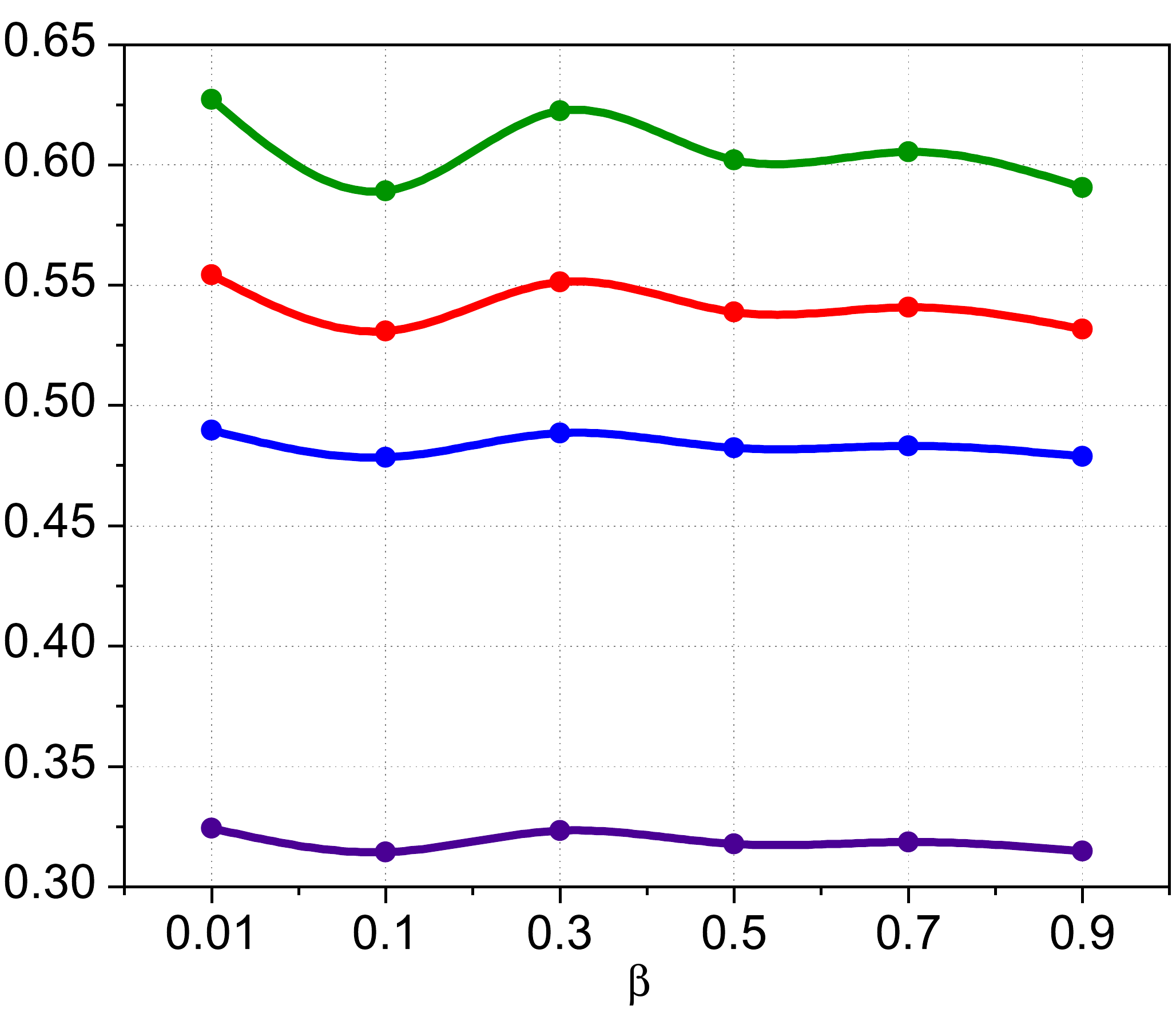}}
    \centerline{(b) SDGE-sum}
  \end{minipage}
  \vfill
\centering
\caption{The experimental results of SDGE with different $\beta$ values.}\label{fg4}
\end{figure*}

\begin{figure*}
  \begin{minipage}{9cm}
    \centerline{\includegraphics[width=7cm]{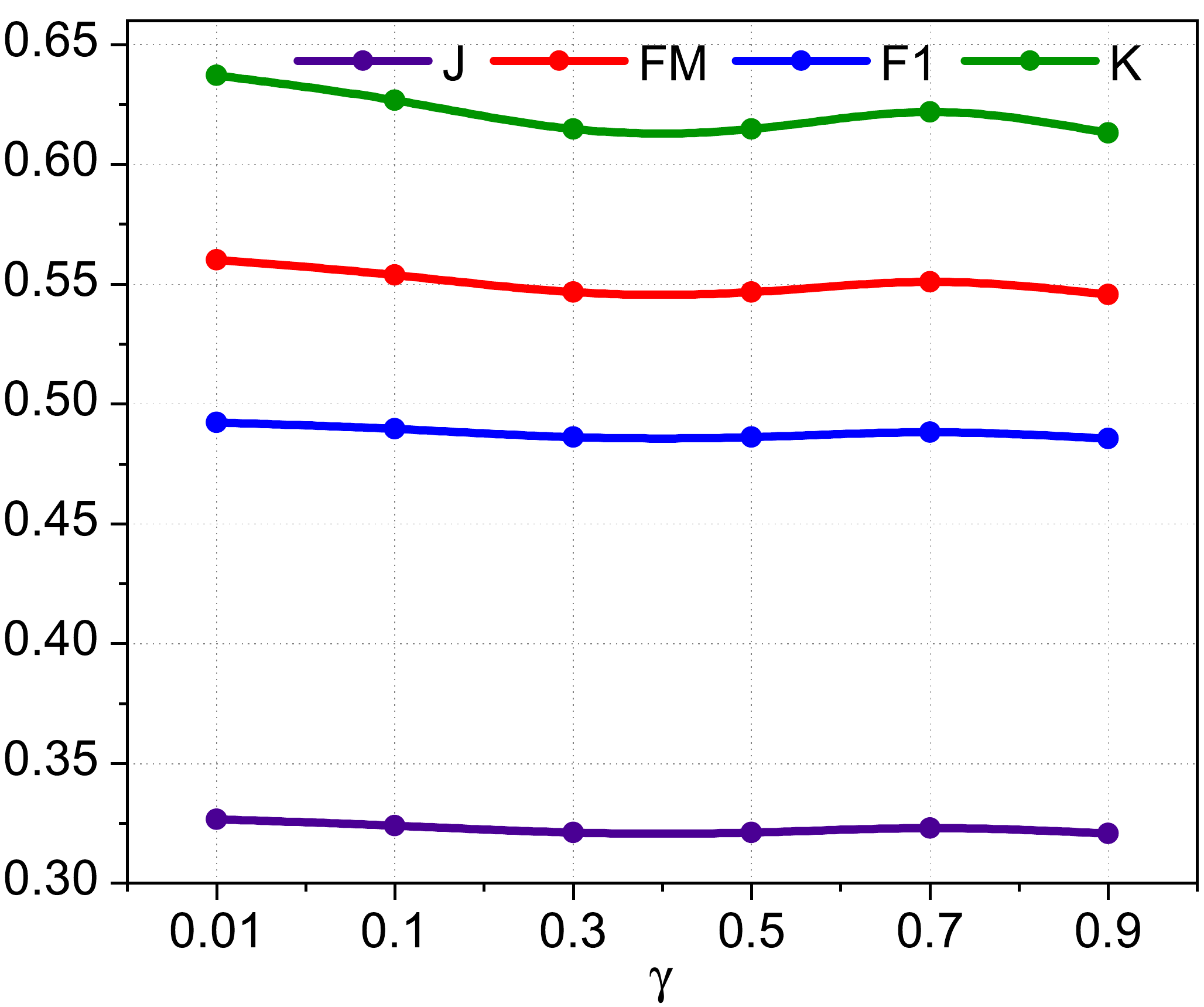}}
    \centerline{(a) SDGE-cat}
  \end{minipage}
  \hfill
  \begin{minipage}{9cm}
    \centerline{\includegraphics[width=7cm]{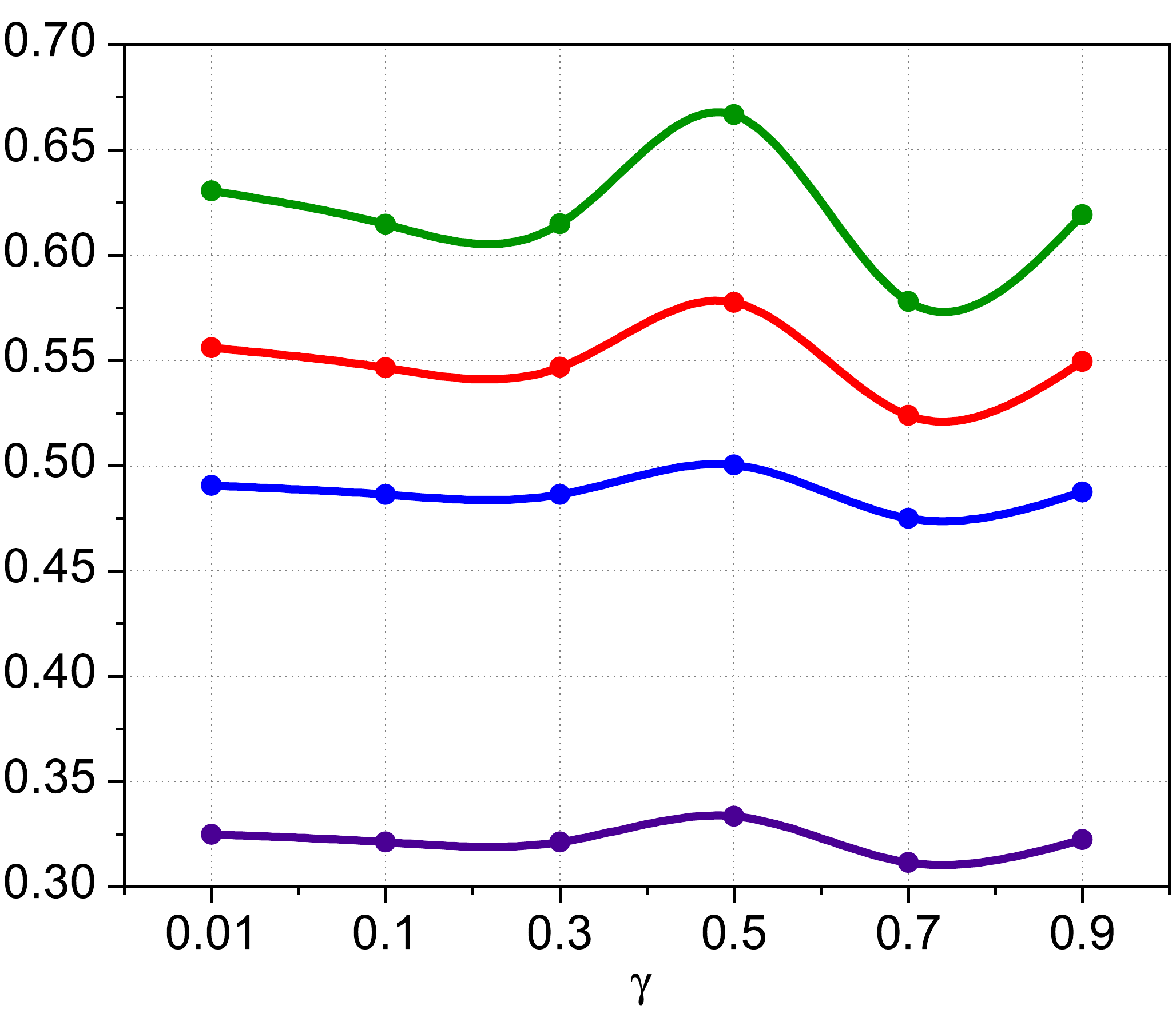}}
    \centerline{(b) SDGE-sum}
  \end{minipage}
  \vfill
\centering
\caption{The experimental results of SDGE with different $\gamma$ values.}\label{fg5}
\end{figure*}

Figs. \ref{fg3}-\ref{fg5} show the experimental results of SDGE with different parameters $\tau$, $\beta$ and $\gamma$ values on ACM data set. From the results, it is known that the parameters $\tau$, $\beta$ and $\gamma$ have influences on the performance of SDGE. For the parameter $\tau$ which is a temperature for the loss of self-supervised learning, it determines the effect of the similarity between two samples on the loss of Eq.(\ref{eq10}). For the parameter $\beta$, it determines the effect of the reconstruction error on the loss of Eq.(\ref{eq10}). For the parameter $\gamma$, it determines the effect of graph regularization on the loss of Eq.(\ref{eq10}). From Figs. \ref{fg3}-\ref{fg5}, the changed trends of SDGE-sum with different $\tau$, $\beta$ and $\gamma$ values are obvious. It means the influence of the parameters $\tau$, $\beta$ and $\gamma$ are also obvious on the performance of SDGE-sum. From Fig. \ref{fg4}(a), it can be seen that the changed trend of SDGE-cat with different $\beta$ values are obvious. In Fig. \ref{fg3}(a), the changed trend of SDGE-cat with different $\tau$ values is much more gradual than SDGE-sum although there is still some changes. It means SDGE-cat is not too sensitive to the parameter $\tau$.

From Fig. \ref{fg5}(a), the changed trend of SDGE-cat with different $\gamma$ values is smooth and the changes of evaluation criteria are small. It means that SDGE-cat is not sensitive to the parameter $\gamma$. From the theory of graph neural network, it is known that the embedding vector of each node comes from the feature vectors of the neighbors which is in tune with manifold learning approaches such as Local Linear Embedding, etc. It means graph neural network can learn the intrinsic manifold structure of graph data and the parameter $\gamma$ can only further increase the effect of graph regularization which has been inherently included in the learning process of graph neural network. Therefore the changed trend of SDGE-cat with different $\gamma$ values is smooth in Figure \ref{fg5}(a) if the GCNs are effectively trained and have good approximation ability.
\subsubsection{Ablative analysis}
It is still not clear if the good performance of SDGE are due to the fusion of high-order information and Dynamic ReLU \cite{chen2020dynamic}. In order to answer this question, the ablation of SDGE is studied in this section. A GCN with four layers which is trained by unsupervised autocoder (denoted as GCN-AE) is as the ablative analysis for the fusion of high-order information. The SDGE algorithm that the Dynamic ReLU activation function is replaced by ReLU activation function (denoted as SDGE-ReLU) is as the ablative analysis for Dynamic ReLU. The experimental algorithms are conducted on ACM data set. The experimental results are showed as Table \ref{tb7} and Figure \ref{fg6}.
\begin{table}[h]
\caption{The experimental results of the ablative analysis ($\tau=100$, $\beta=1$, $\gamma=1$ and \emph{epoch}=100).}
\centering
\begin{tabular}{ccccccc}\toprule
Algorithms & J & FM & $F_{1}$ & K \\\midrule
LLE &0.2634 &0.4366 &0.4156 &0.4597\\
GCN-AE &0.2025  &0.3368  &0.3368  &0.3368 \\\hline
SDGE-ReLU-cat &0.3186  &0.5399  &0.4832  &0.6033\\
SDGE-cat &\textbf{0.3222} &\textbf{0.5472} &\textbf{0.4874} &\textbf{0.6143} \\\hline
SDGE-ReLU-sum &0.3176 &0.5382 &0.4820 &0.6010 \\
SDGE-sum &0.3178  &0.5382  &0.4823  &0.6005\\\toprule
\end{tabular}\label{tb7}
\end{table}
\begin{figure*}
  \begin{minipage}{6cm}
    \centerline{\includegraphics[width=6cm]{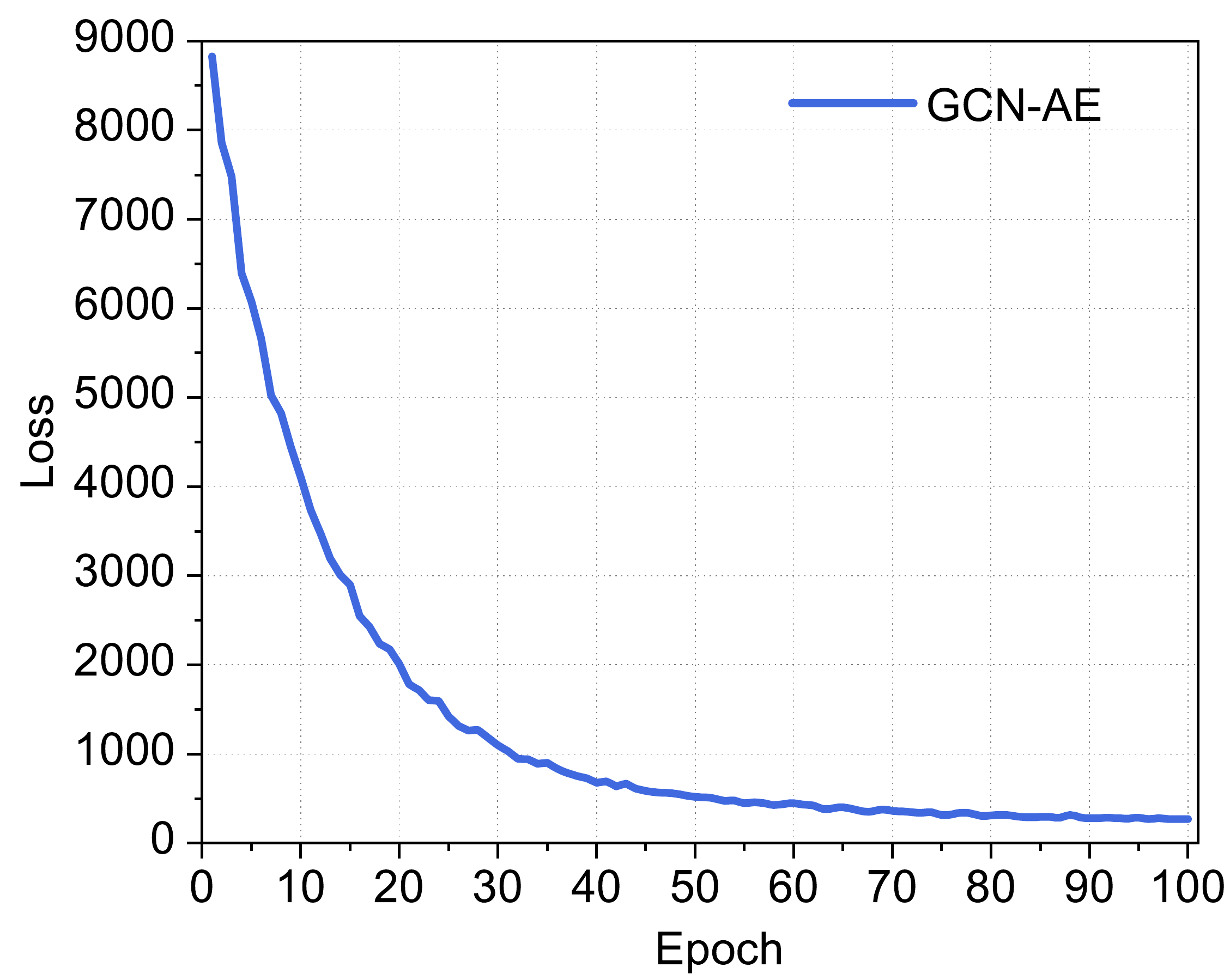}}
    \centerline{(a) GCN-AE}
  \end{minipage}
  \hfill
  \begin{minipage}{6cm}
    \centerline{\includegraphics[width=6cm]{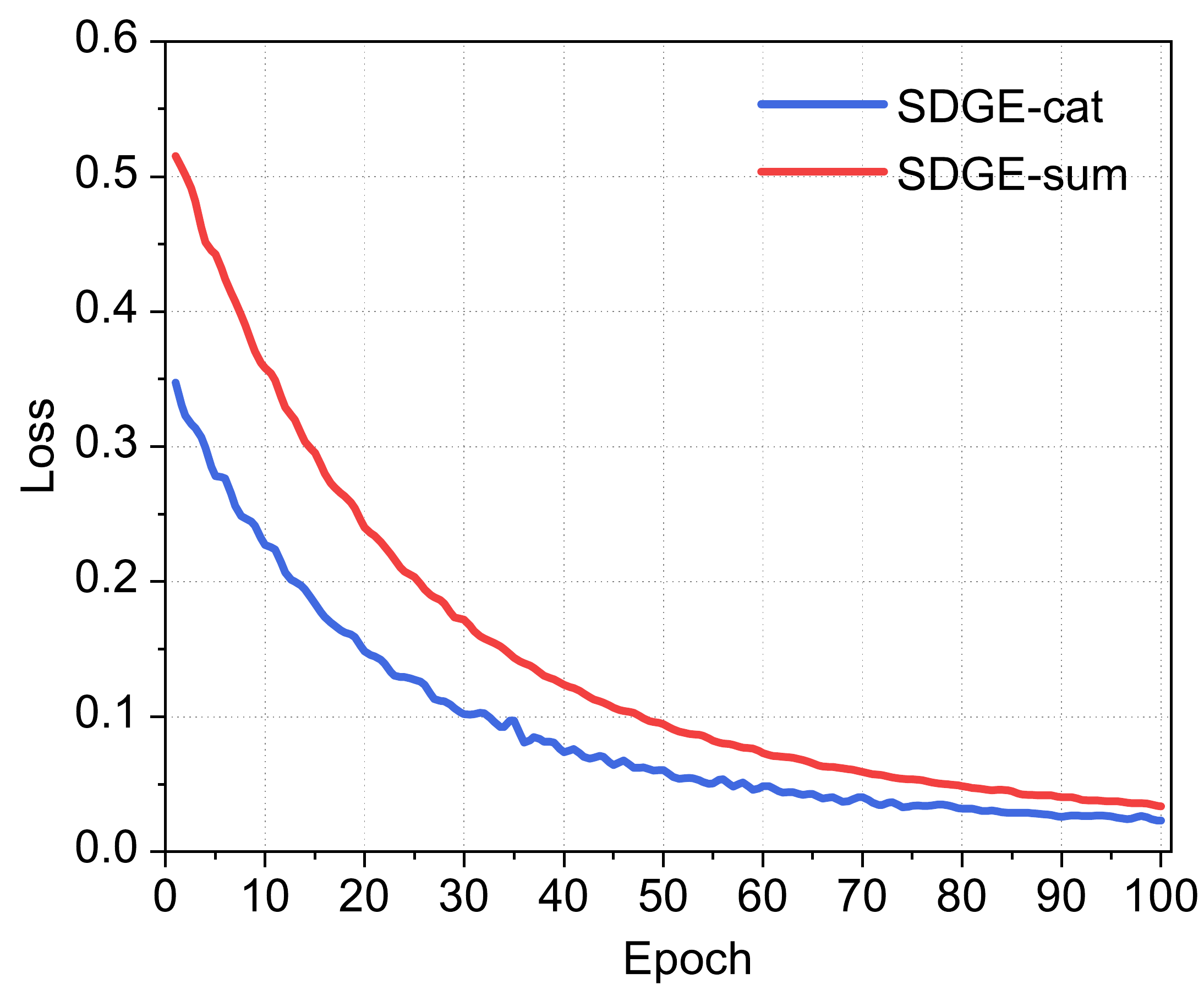}}
    \centerline{(b) SDGE}
  \end{minipage}
  \hfill
  \begin{minipage}{6cm}
    \centerline{\includegraphics[width=6cm]{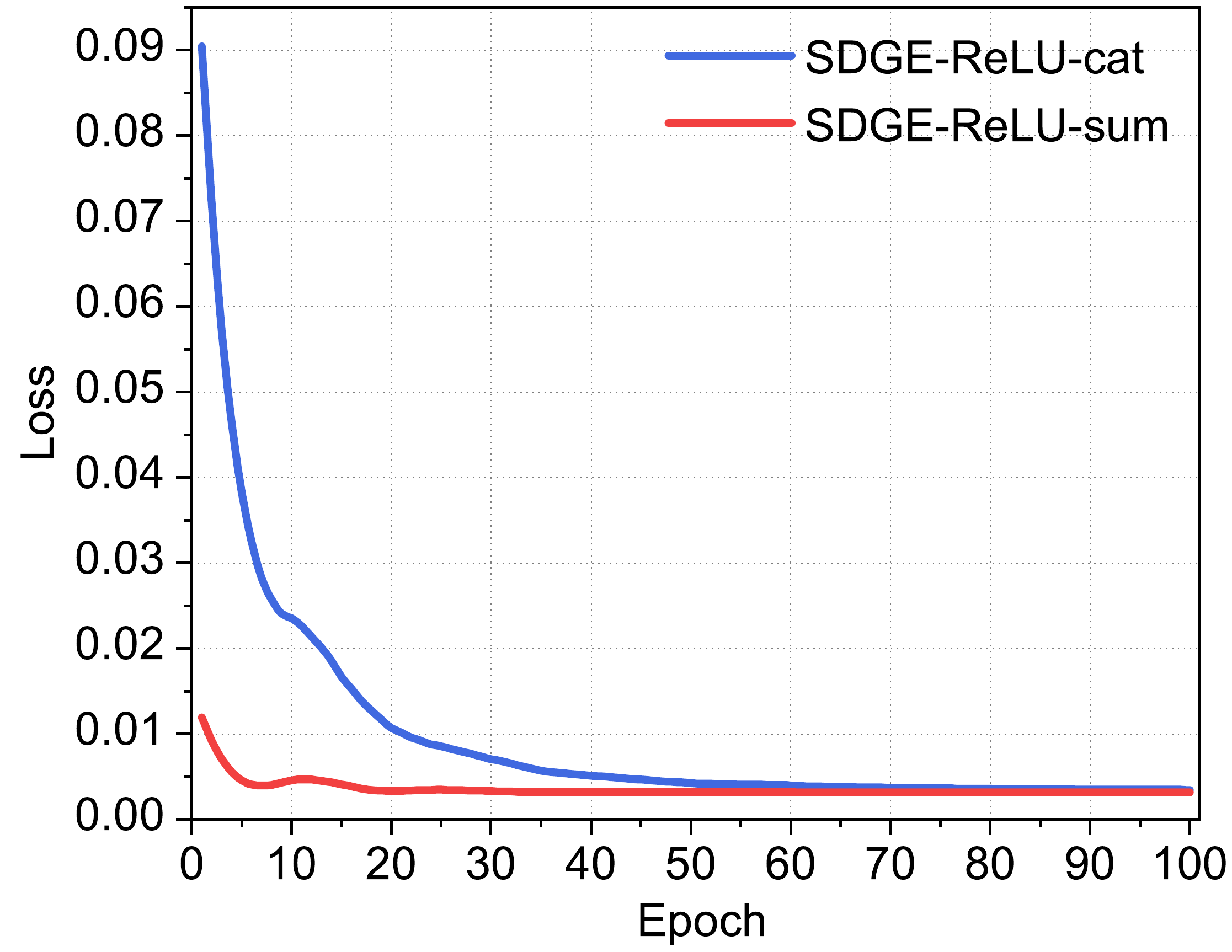}}
    \centerline{(c) SDGE-ReLU}
  \end{minipage}
  \vfill
\centering
\caption{The convergence of GCN-AE, SGDE and SGDE-ReLU.}\label{fg6}
\end{figure*}

Table \ref{tb7} shows the experimental results of the ablative analysis. In order to further confirm the effectiveness of the proposed algorithm, the experimental results of LLE which is a representative manifold learning algorithm is also presented in Table \ref{tb7}. From Table 7, it is known that SDGE-cat is the best of all. It demonstrates that SDGE is an effective approach for community discover. By comparing the results of SDGE-ReLU-cat, SDGE-cat, SDGE-ReLU-sum and SDGE-sum with the results of GCN-AE, it can be seen that the performance of SDGE outperforms GCN-AE and the advantages of the performance between SDGE and GCN-AE are obvious. The obvious advantages indicate the high-order information fusion of SDGE is effective and the fusion of the multiple outputs of GCNs can better learn the local and global structure information of a graph. In SDGE, each GCN can learn a kind of structure information, therefore SDGE can learn the graph information from multiple views and obtain the best results on the experimental data set.

In Table \ref{tb7}, SDGE-cat outperforms SDGE-ReLU-cat. The result demonstrates that the performance of SDGE can be improved by introducing Dynamic ReLU activation function. However, the difference of the performance between SDGE-cat and SDGE-ReLU-cat is small and the results of SDGE-ReLU-sum are almost equal to the results of SDGE-sum. From Fig. \ref{fg6}(b)-(c), it can be seen that SDGE-ReLU-cat and SDGE-ReLU-sum are fully convergent when \emph{epoch} is 100 and the losses of SDGE-cat and SDGE-sum can be future decreased although \emph{epoch} has reached 100. The trend of the loss in Fig. \ref{fg6}(b) means the performance of SDGE-cat and SDGE-sum can be future improved. In the current situation (\emph{epoch}=100), SDGE has already outperformed or be equal to SDGE-ReLU. The performance of SDGE will outperform SDGE-ReLU if SDGE is fully convergent. Therefore the results prove that Dynamic ReLU can improve the performance of SDGE.

Fig. \ref{fg6} shows the convergence of GCN-AE, SDGE and SDGE-ReLU. From the results, it can be seen that GCN-AE is fully convergent when \emph{epoch} is 80, SDGE-ReLU is fully convergent when epoch is 60 and the loss of SDGE can be future decreased after \emph{epoch} is larger than 100. The results also show that SDGE requires more iterations than GCN-AE and SDGE-ReLU. In SDGE, it needs to train multiple GCNs which increases the complexity of the algorithm. In addition, SDGE introduces Dynamic ReLU as activation function. From the reference \cite{chen2020dynamic}, it is known that the appropriate parameters of Dynamic ReLU are determined according to the global context which needs more computation. Therefore SDGE requires more iterations to reach the convergence.

\begin{table}[h]
\caption{The experimental results of the ablative analysis of spectral propagation (SP) on Hyperplane data set ($\tau=50$, $\beta=0.5$, $\gamma=0.5$ and \emph{epoch}=50).}
\centering
\begin{tabular}{cccccc}\toprule
Algorithms & J & FM & $F_{1}$ & K \\\midrule
SDGE-cat &\textbf{0.5000}  &\textbf{0.7071}  &\textbf{0.6667}  &\textbf{0.7500} \\
SDGE-cat-w/o-SP &0.4514  &0.6412  &0.6220  &0.6610 \\\hline
SDGE-sum &\textbf{0.5000}  &\textbf{0.7071}  &\textbf{0.6667}  &\textbf{0.7500}\\
SDGE-sum-w/o-SP&0.4974  &0.7034  &0.6643  &0.7448 \\\toprule
\end{tabular}\label{tb8}
\end{table}

\begin{figure}
  \centering
  \includegraphics[width=8cm]{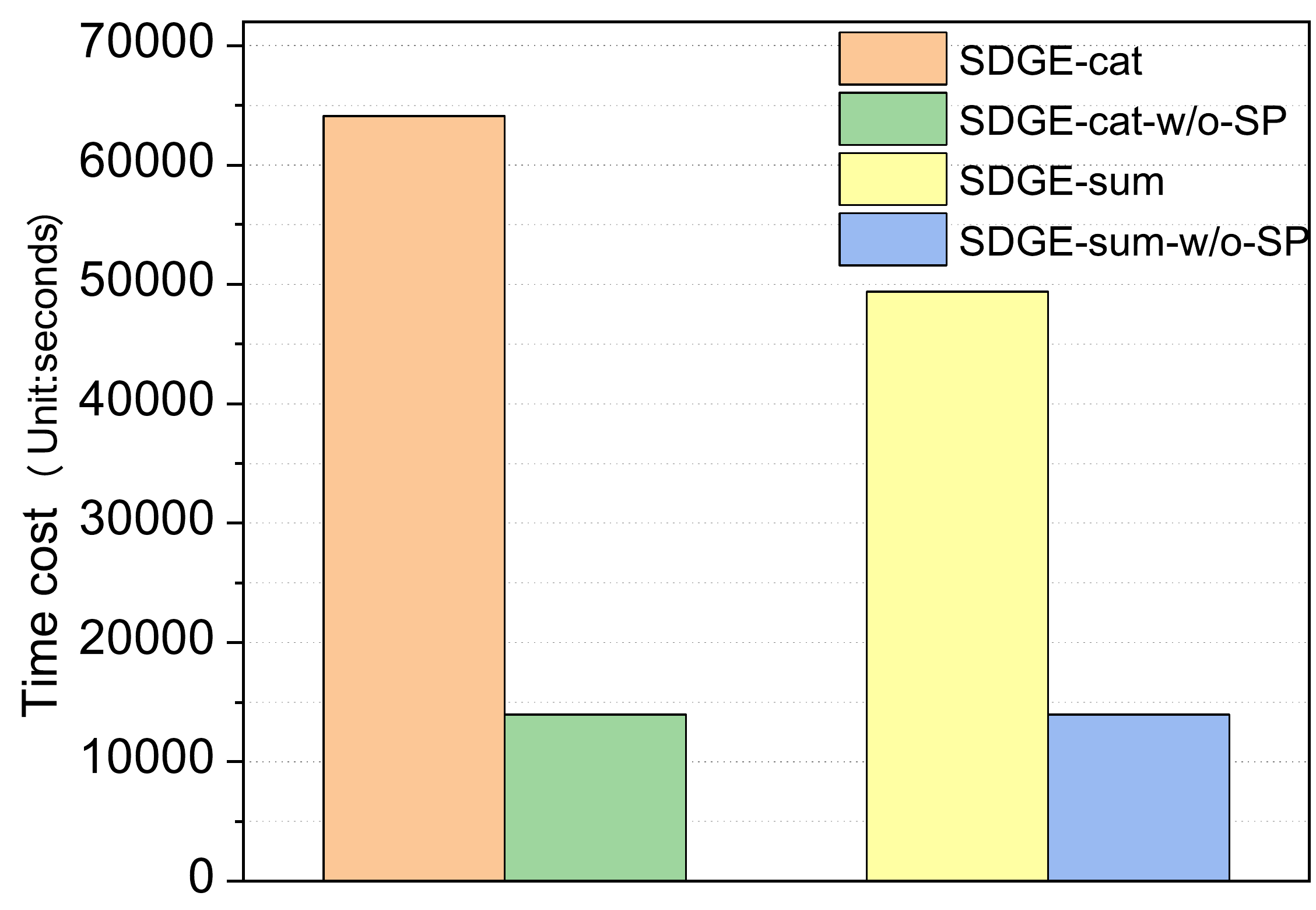}\\
  \caption{The time cost of the ablative analysis of spectral propagation (SP) on Hyperplane data set}\label{fg7}
\end{figure}

Table \ref{tb8} shows the ablative result of spectral propagation on Hyperplane data set. From the results, it is known that SDGE-cat and SDGE-sum outperform SDGE-cat-w/o-SP and SDGE-sum-w/o-SP, respectively. It demonstrates that spectral propagation improves the performance of SDGE-cat. When aggregate function is $CONCAT\left ( \cdot  \right )$, the performance degradation is obvious after spectral propagation is removed from SDGE-cat. The result provides evidence that spectral propagation plays an important role in SDGE-cat. However, the performance difference between SDGE-sum and SDGE-sum-w/o-SP is small although the performance degradation of SDGE-sum also appears after spectral propagation is removed from SDGE-sum. Therefore it can conclude that spectral propagation has more significant influence on SDGE-cat than SDGE-sum.

Fig. \ref{fg7} shows the time cost of the ablative analysis about spectral propagation on Hyperplane data set. From the results, it is known that both SDGE-cat and SDGE-sum spend more time on the calculation of embedding vectors than SDGE-cat-w/o-SP and SDGE-sum-w/o-SP. It means spectral propagation is a time consuming process. However, the results in Table \ref{tb8} show the effectiveness of spectral propagation. Therefore  spectral propagation is an effective approach to improve the performance of SDGE if the task is sensitive to time cost.
\subsubsection{Efficiency analysis}
In this section, the efficiency of SGDE is studied. Figure \ref{fg8} shows the time cost of SGDE algorithm on the experimental data sets.
\begin{figure*}
  \begin{minipage}{6cm}
    \centerline{\includegraphics[width=5.5cm]{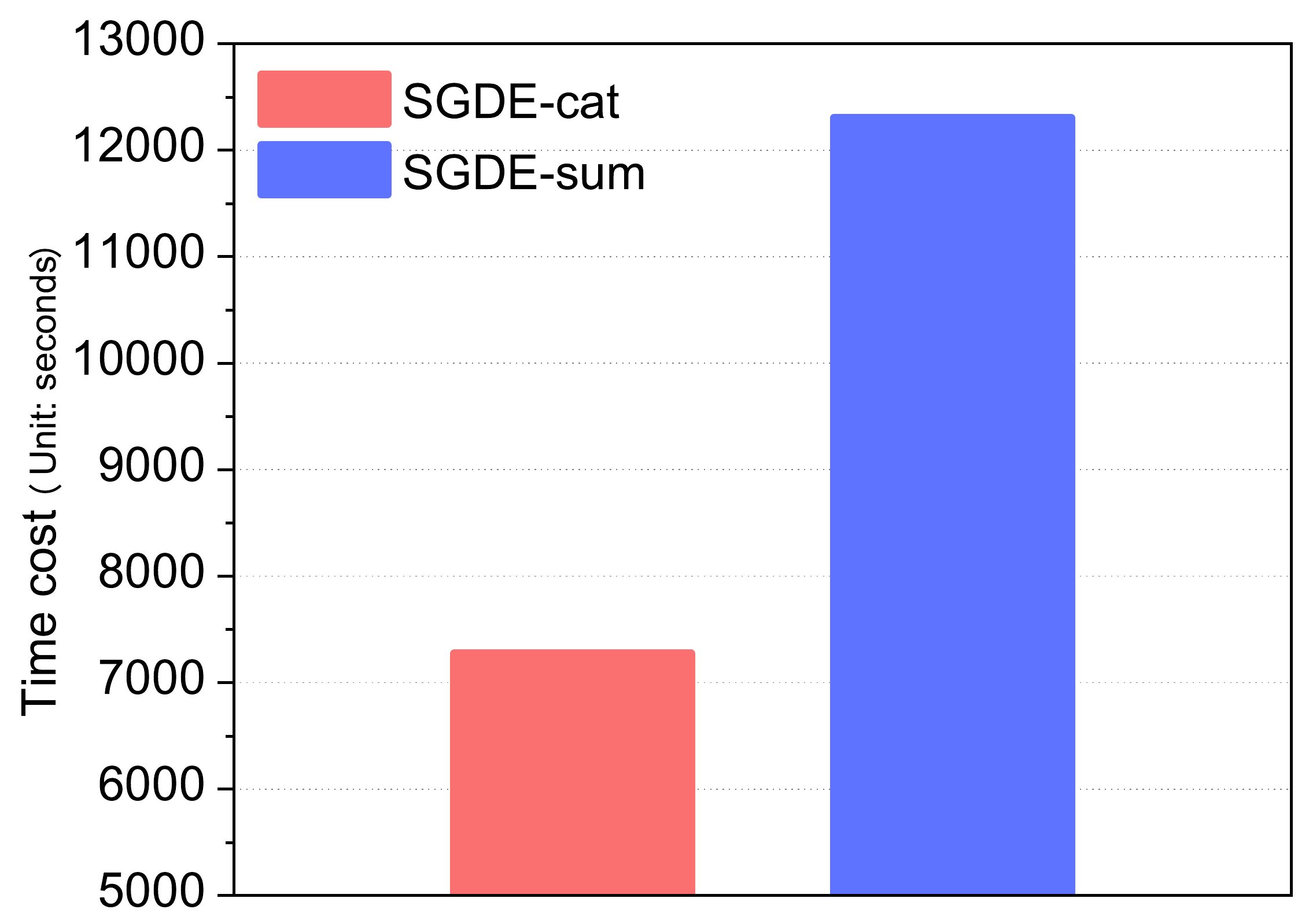}}
    \centerline{(a) ACM data set.}
  \end{minipage}
  \hfill
  \begin{minipage}{6cm}
    \centerline{\includegraphics[width=5.5cm]{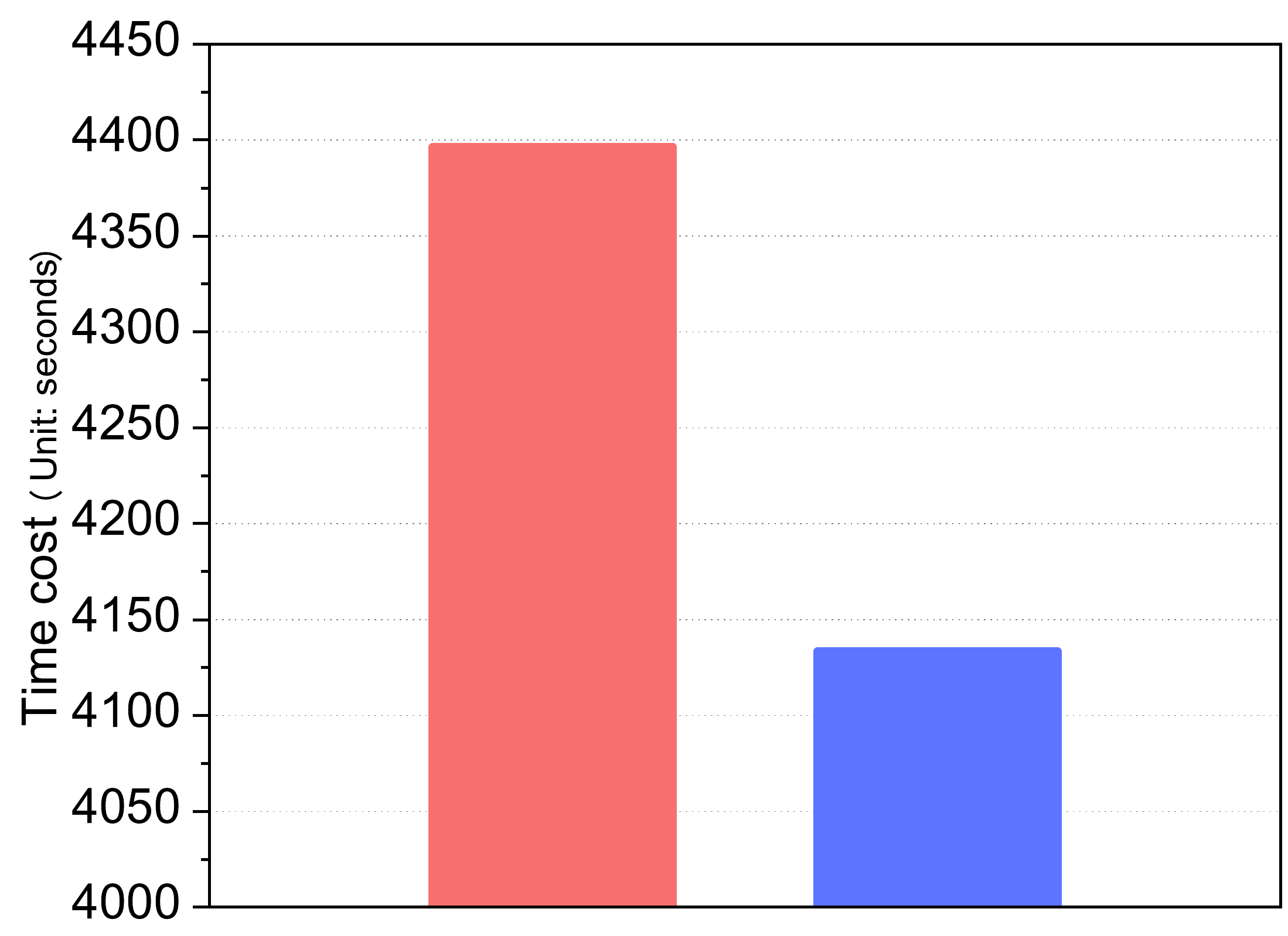}}
    \centerline{(b) USA data set.}
  \end{minipage}
  \hfill
  \begin{minipage}{6cm}
    \centerline{\includegraphics[width=5.5cm]{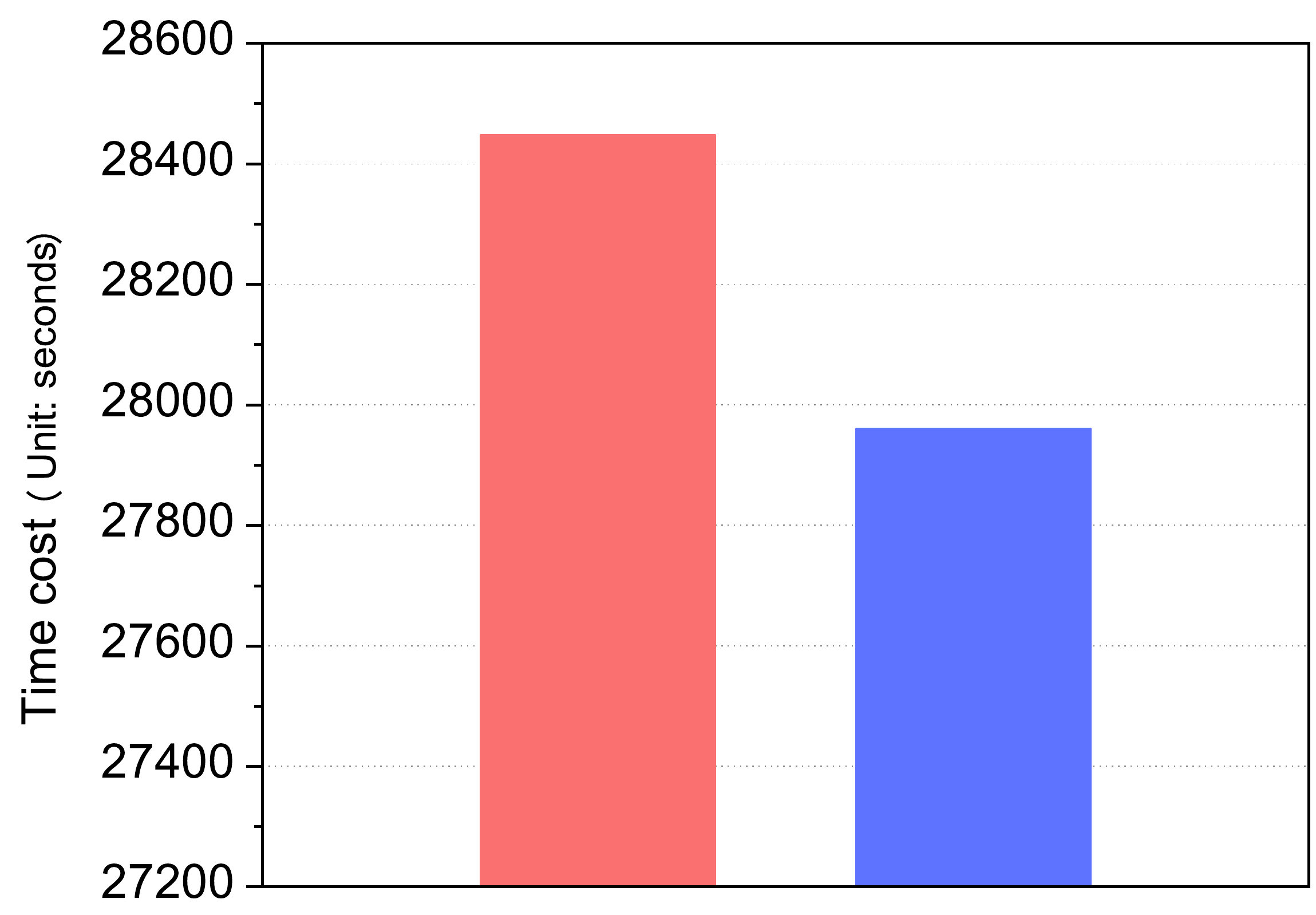}}
    \centerline{(c) Image data set.}
  \end{minipage}
  \vfill
  \begin{minipage}{6cm}
    \centerline{\includegraphics[width=5.5cm]{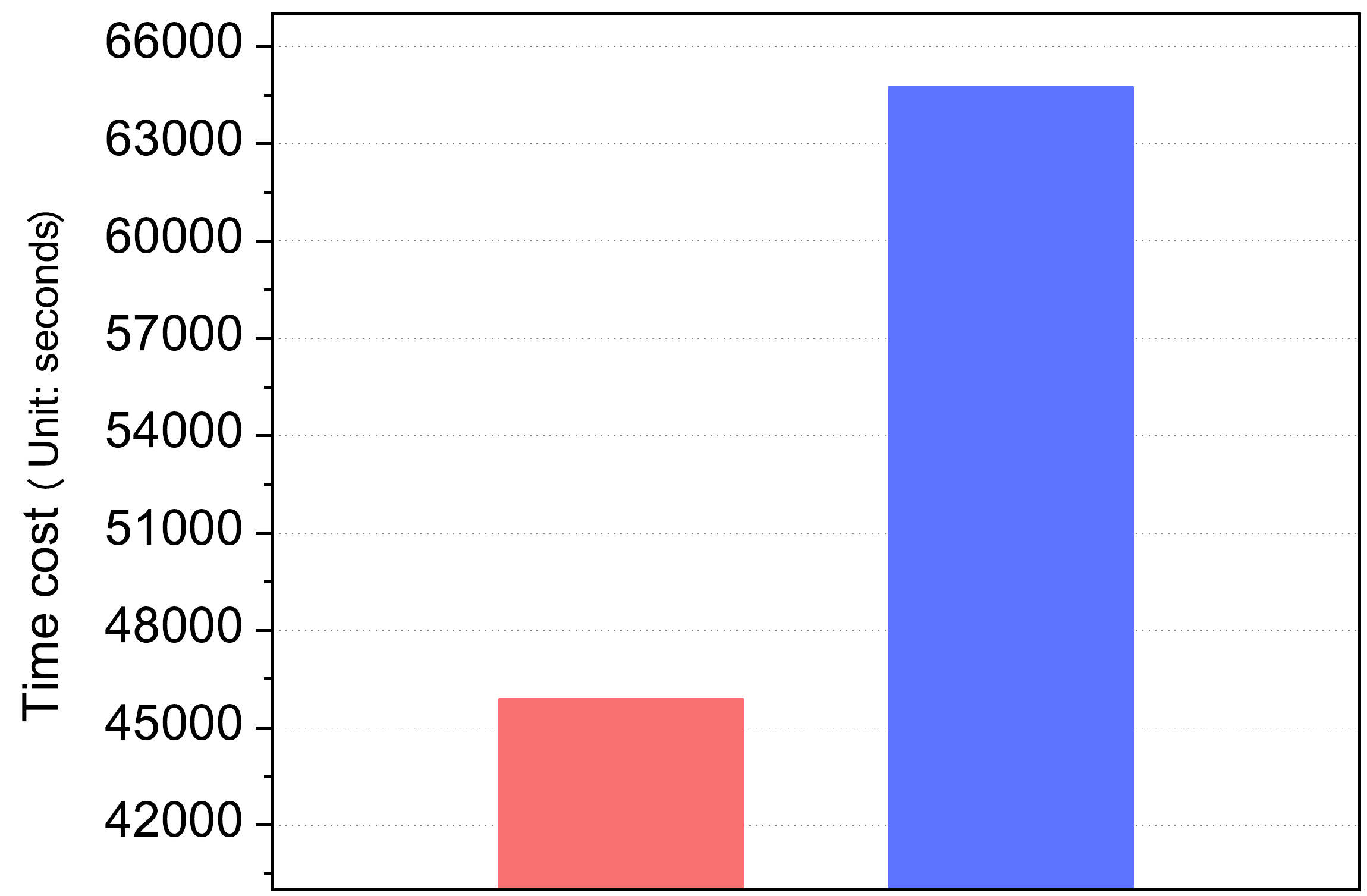}}
    \centerline{(d) Hyperplane data set.}
  \end{minipage}
  \hfill
  \begin{minipage}{6cm}
    \centerline{\includegraphics[width=5.5cm]{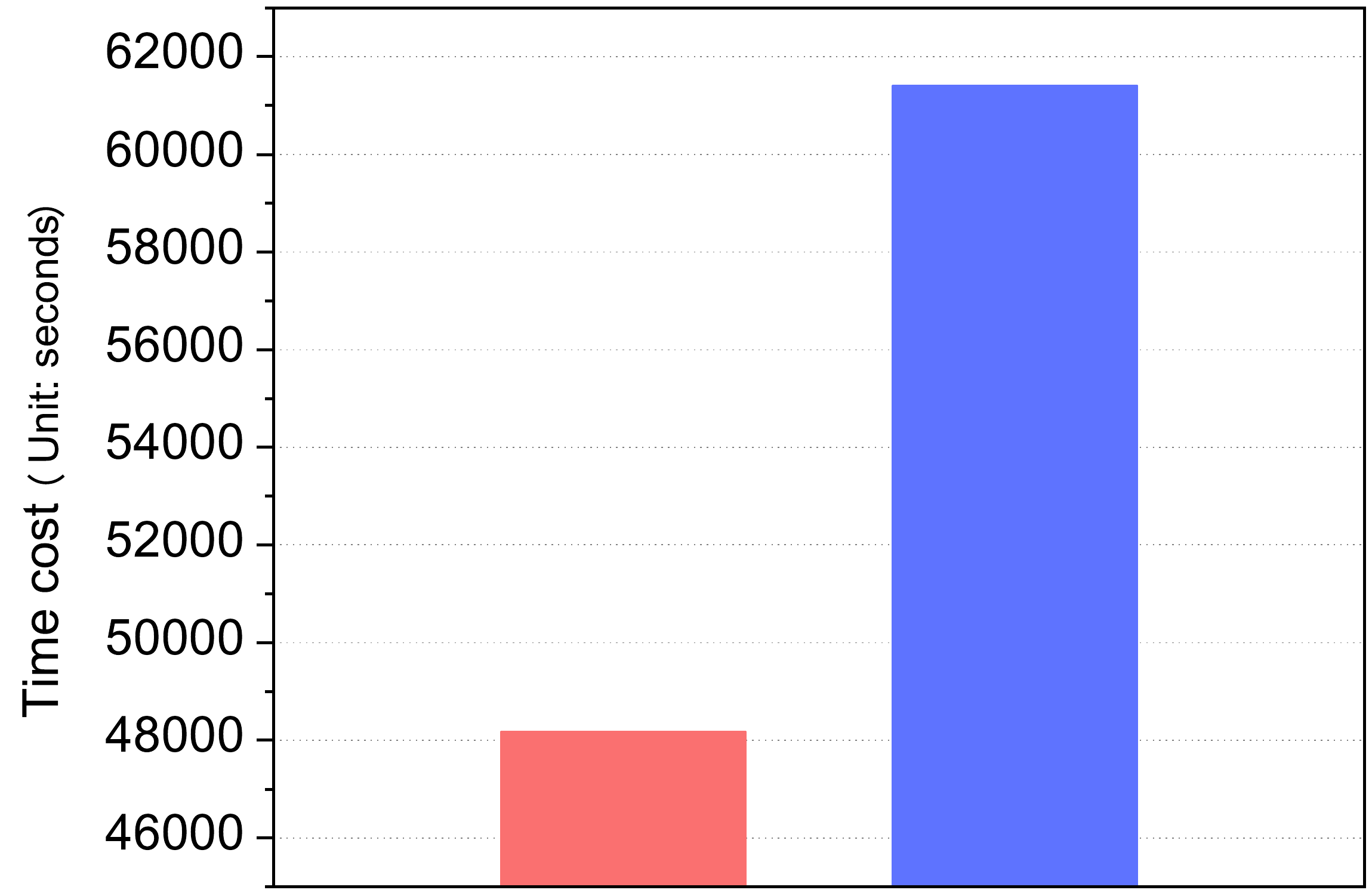}}
    \centerline{(e) Waveform data set.}
  \end{minipage}
  \vfill
\caption{The time cost of SGDE on the experimental data sets.}\label{fg8}
\end{figure*}

From the results, it can be seen that time cost of SDGE-sum is more than SGDE-cat on the most experimental data sets. It known that the outputs of GCNs are added if the aggregate function is $sum\left ( \cdot  \right )$ and the time complexity is $\mathcal{O}\left ( nr \right )$. If the aggregate function is $ CONCAT\left ( \cdot  \right )$, the outputs of GCNs are concatenated together and the time complexity is $\mathcal{O}\left ( rd_{r} \right )$ where $d_{r}$ is the output dimension of GCN. Due to the fact $(n\gg d_{r})$, the time complexity of $sum\left ( \cdot  \right )$ is higher than $CONCAT\left ( \cdot  \right )$ which is consistent with the results in Figs. \ref{fg8} except for Figs. \ref{fg8}(b)-(c). In Figs. \ref{fg8}(b)-(c), the time cost of $sum\left ( \cdot  \right )$ is less than $CONCAT\left ( \cdot  \right )$. For USA and Image data sets, the number of nodes \emph{n} is not too large and it means the difference between \emph{n} and $d_{r}$ is smaller than the other data sets. Therefore the time cost of $sum\left ( \cdot  \right )$ is closer to the time cost of $CONCAT\left ( \cdot  \right )$ on USA and Image data sets than the other data sets and the changed range in Figs. \ref{fg8}(b)-(c) also agrees with the above analysis. In addition, the initialization of graph neural network also plays an important role on the final time cost. An good initialization can speed up the convergence. Therefore the time cost of $CONCAT\left ( \cdot  \right )$ is more than  $sum\left ( \cdot  \right )$ on USA and Image data sets.
\section{Conclusions}
In this paper, we focus on building a self-supervised deep graph neural network named SDGE for node embedding and community discovery. Through the fusion of the outputs of multiple GCNs, SDGE can effectively utilize the high-order information of graph. The embedding result of SDGE can preserve the structure and node similarity in the low dimension embedding space. The spectral propagation is also introduced to enhance the embedding result. The extensive experiments on the experimental data sets demonstrate the effectiveness of SDGE algorithm.

In the current, pre-training has been proved that it is an effective approach to improve the performance of graph neural network. Therefore the future work would suggest to explore to introduce pre-training approach for SDGE. Moreover, the depth of neural network is important for the representation learning ability of GCN and the increasing depth of GCN is also a potential way of improving the proposed algorithm.

\ifCLASSOPTIONcompsoc
\section*{Acknowledgments}
This work was supported by National Natural Science Fund of China (Nos. 61972064, 61672130), Liaoning Revitalization Talents Program (No. XLY-C1806006), Fundamental Research Funds for the Central Universities (No. DUT19RC(3)012).
\fi

\ifCLASSOPTIONcaptionsoff
  \newpage
\fi

\bibliographystyle{IEEEtran}
\bibliography{acmart}

\begin{thebibliography}{10}
\providecommand{\url}[1]{#1}
\csname url@samestyle\endcsname
\providecommand{\newblock}{\relax}
\providecommand{\bibinfo}[2]{#2}
\providecommand{\BIBentrySTDinterwordspacing}{\spaceskip=0pt\relax}
\providecommand{\BIBentryALTinterwordstretchfactor}{4}
\providecommand{\BIBentryALTinterwordspacing}{\spaceskip=\fontdimen2\font plus
\BIBentryALTinterwordstretchfactor\fontdimen3\font minus
  \fontdimen4\font\relax}
\providecommand{\BIBforeignlanguage}[2]{{%
\expandafter\ifx\csname l@#1\endcsname\relax
\typeout{** WARNING: IEEEtran.bst: No hyphenation pattern has been}%
\typeout{** loaded for the language `#1'. Using the pattern for}%
\typeout{** the default language instead.}%
\else
\language=\csname l@#1\endcsname
\fi
#2}}
\providecommand{\BIBdecl}{\relax}
\BIBdecl

\bibitem{tang2017computational}
J.~Tang, ``Computational models for social network analysis: A brief survey,''
  in \emph{Proceedings of the 26th International Conference on World Wide Web
  Companion}, 2017, pp. 921--925.

\bibitem{liu2020deep}
F.~Liu, S.~Xue, J.~Wu, C.~Zhou, W.~Hu, C.~Paris, S.~Nepal, J.~Yang, and P.~S.
  Yu, ``Deep learning for community detection: Progress, challenges and
  opportunities,'' in \emph{The 29th International Joint Conference on
  Artificial Intelligence}, 2020, pp. 4981--4987.

\bibitem{yang2015defining}
J.~Yang and J.~Leskovec, ``Defining and evaluating network communities based on
  ground-truth,'' \emph{Knowledge and Information Systems}, vol.~42, no.~1, pp.
  181--213, 2015.

\bibitem{du2019sequential}
Z.~Du, X.~Wang, H.~Yang, J.~Zhou, and J.~Tang, ``Sequential scenario-specific
  meta learner for online recommendation,'' in \emph{Proceedings of the 25th
  ACM SIGKDD International Conference on Knowledge Discovery and Data Mining},
  2019, pp. 2895--2904.

\bibitem{cen2020controllable}
Y.~Cen, J.~Zhang, X.~Zou, C.~Zhou, H.~Yang, and J.~Tang, ``Controllable
  multi-interest framework for recommendation,'' in \emph{Proceedings of the
  26th ACM SIGKDD International Conference on Knowledge Discovery and Data
  Mining}, 2020, pp. 2942--2951.

\bibitem{zhang2019iteratively}
W.~Zhang, B.~Paudel, L.~Wang, J.~Chen, H.~Zhu, W.~Zhang, A.~Bernstein, and
  H.~Chen, ``Iteratively learning embeddings and rules for knowledge graph
  reasoning,'' in \emph{The World Wide Web Conference}, 2019, pp. 2366--2377.

\bibitem{liu2019integrating}
L.~Liu, Y.~Ma, X.~Zhu, Y.~Yang, X.~Hao, L.~Wang, and J.~Peng, ``Integrating
  sequence and network information to enhance protein-protein interaction
  prediction using graph convolutional networks,'' in \emph{2019 IEEE
  International Conference on Bioinformatics and Biomedicine}.\hskip 1em plus
  0.5em minus 0.4em\relax IEEE, 2019, pp. 1762--1768.

\bibitem{cui2018survey}
P.~Cui, X.~Wang, J.~Pei, and W.~Zhu, ``A survey on network embedding,''
  \emph{IEEE Transactions on Knowledge and Data Engineering}, vol.~31, no.~5,
  pp. 833--852, 2018.

\bibitem{roweis2000nonlinear}
S.~T. Roweis and L.~K. Saul, ``Nonlinear dimensionality reduction by locally
  linear embedding,'' \emph{Science}, vol. 290, no. 5500, pp. 2323--2326, 2000.

\bibitem{belkin2002laplacian}
M.~Belkin and P.~Niyogi, ``Laplacian eigenmaps and spectral techniques for
  embedding and clustering,'' in \emph{Advances in Neural Information
  Processing Systems}, 2002, pp. 585--591.

\bibitem{lee1999learning}
D.~D. Lee and H.~S. Seung, ``Learning the parts of objects by non-negative
  matrix factorization,'' \emph{Nature}, vol. 401, no. 6755, pp. 788--791,
  1999.

\bibitem{rong2020deep}
Y.~Rong, T.~Xu, J.~Huang, W.~Huang, H.~Cheng, Y.~Ma, Y.~Wang, T.~Derr, L.~Wu,
  and T.~Ma, ``Deep graph learning: Foundations, advances and applications,''
  in \emph{Proceedings of the 26th ACM SIGKDD International Conference on
  Knowledge Discovery and Data Mining}, 2020, pp. 3555--3556.

\bibitem{qiu2019netsmf}
J.~Qiu, Y.~Dong, H.~Ma, J.~Li, C.~Wang, K.~Wang, and J.~Tang, ``Netsmf:
  Large-scale network embedding as sparse matrix factorization,'' in \emph{The
  World Wide Web Conference}, 2019, pp. 1509--1520.

\bibitem{pang2017flexible}
T.~Pang, F.~Nie, and J.~Han, ``Flexible orthogonal neighborhood preserving
  embedding.'' in \emph{The Twenty-Sixth International Joint Conference on
  Artificial Intelligence}, 2017, pp. 2592--2598.

\bibitem{zhang2018arbitrary}
Z.~Zhang, P.~Cui, X.~Wang, J.~Pei, X.~Yao, and W.~Zhu, ``Arbitrary-order
  proximity preserved network embedding,'' in \emph{Proceedings of the 24th ACM
  SIGKDD International Conference on Knowledge Discovery and Data Mining},
  2018, pp. 2778--2786.

\bibitem{wang2017community}
X.~Wang, P.~Cui, J.~Wang, J.~Pei, W.~Zhu, and S.~Yang, ``Community preserving
  network embedding.'' in \emph{The Thirty-First AAAI Conference on Artificial
  Intelligence}, 2017, pp. 203--209.

\bibitem{qiu2019noise}
Z.~Qiu, W.~Hu, J.~Wu, Z.~Tang, and X.~Jia, ``Noise-resilient similarity
  preserving network embedding for social networks.'' in \emph{The
  Twenty-Eighth International Joint Conference on Artificial Intelligence},
  2019, pp. 3282--3288.

\bibitem{mikolov2013efficient}
T.~Mikolov, K.~Chen, G.~Corrado, and J.~Dean, ``Efficient estimation of word
  representations in vector space,'' \emph{arXiv preprint arXiv:1301.3781},
  2013.

\bibitem{perozzi2014deepwalk}
B.~Perozzi, R.~Al-Rfou, and S.~Skiena, ``Deepwalk: Online learning of social
  representations,'' in \emph{Proceedings of the 20th ACM SIGKDD International
  Conference on Knowledge Discovery and Data Mining}, 2014, pp. 701--710.

\bibitem{tang2015line}
J.~Tang, M.~Qu, M.~Wang, M.~Zhang, J.~Yan, and Q.~Mei, ``Line: Large-scale
  information network embedding,'' in \emph{Proceedings of the 24th
  International Conference on World Wide Web}, 2015, pp. 1067--1077.

\bibitem{ribeiro2017struc2vec}
L.~F. Ribeiro, P.~H. Saverese, and D.~R. Figueiredo, ``struc2vec: Learning node
  representations from structural identity,'' in \emph{Proceedings of the 23rd
  ACM SIGKDD International Conference on Knowledge Discovery and Data Mining},
  2017, pp. 385--394.

\bibitem{dong2017metapath2vec}
Y.~Dong, N.~V. Chawla, and A.~Swami, ``metapath2vec: Scalable representation
  learning for heterogeneous networks,'' in \emph{Proceedings of the 23rd ACM
  SIGKDD International Conference on Knowledge Discovery and Data Mining},
  2017, pp. 135--144.

\bibitem{wang2019inductive}
L.~Wang, B.~Zong, Q.~Ma, W.~Cheng, J.~Ni, W.~Yu, Y.~Liu, D.~Song, H.~Chen, and
  Y.~Fu, ``Inductive and unsupervised representation learning on graph
  structured objects,'' in \emph{International Conference on Learning
  Representations}, 2020.

\bibitem{guo2019spine}
J.~Guo, L.~Xu, and J.~Liu, ``Spine: Structural identity preserved inductive
  network embedding,'' in \emph{The Twenty-Eighth International Joint
  Conference on Artificial Intelligence}, 2019, pp. 2399--2405.

\bibitem{jin2019incorporating}
D.~Jin, X.~You, W.~Li, D.~He, P.~Cui, F.~Fogelman-Souli{\'e}, and
  T.~Chakraborty, ``Incorporating network embedding into markov random field
  for better community detection,'' in \emph{The Thirty-Third AAAI Conference
  on Artificial Intelligence}, 2019, pp. 160--167.

\bibitem{kipf2016semi}
T.~N. Kipf and M.~Welling, ``Semi-supervised classification with graph
  convolutional networks,'' in \emph{International Conference on Learning
  Representations}, 2017.

\bibitem{wu2020comprehensive}
Z.~Wu, S.~Pan, F.~Chen, G.~Long, C.~Zhang, and S.~Y. Philip, ``A comprehensive
  survey on graph neural networks,'' \emph{IEEE Transactions on Neural Networks
  and Learning Systems}, 2020.

\bibitem{xu2020powerful}
K.~Xu, W.~Hu, J.~Leskovec, and S.~Jegelka, ``How powerful are graph neural
  networks?'' in \emph{International Conference on Learning Representations},
  2020.

\bibitem{loukas2020graph}
A.~Loukas, ``What graph neural networks cannot learn: depth vs width,'' in
  \emph{International Conference on Learning Representations}, 2020.

\bibitem{chen2020simple}
M.~Chen, Z.~Wei, Z.~Huang, B.~Ding, and Y.~Li, ``Simple and deep graph
  convolutional networks,'' in \emph{Proceedings of the 37th International
  Conference on Machine Learning}, 2020.

\bibitem{hamilton2017inductive}
W.~Hamilton, Z.~Ying, and J.~Leskovec, ``Inductive representation learning on
  large graphs,'' in \emph{Advances in Neural Information Processing Systems},
  2017, pp. 1024--1034.

\bibitem{rong2020dropedge}
Y.~Rong, W.~Huang, T.~Xu, and J.~Huang, ``Dropedge: Towards deep graph
  convolutional networks on node classification,'' in \emph{International
  Conference on Learning Representations}, 2020.

\bibitem{chen2020convolutional}
D.~Chen, L.~Jacob, and J.~Mairal, ``Convolutional kernel networks for
  graph-structured data,'' in \emph{Proceedings of the 37th International
  Conference on Machine Learning}, 2020.

\bibitem{peng2020graph}
Z.~Peng, W.~Huang, M.~Luo, Q.~Zheng, Y.~Rong, T.~Xu, and J.~Huang, ``Graph
  representation learning via graphical mutual information maximization,'' in
  \emph{Proceedings of The Web Conference 2020}, 2020, pp. 259--270.

\bibitem{bo2020structural}
D.~Bo, X.~Wang, C.~Shi, M.~Zhu, E.~Lu, and P.~Cui, ``Structural deep clustering
  network,'' in \emph{Proceedings of The Web Conference 2020}, 2020, pp.
  1400--1410.

\bibitem{zhangdropping2020}
X.~Zhang, C.~Xu, and D.~Tao, ``On dropping clusters to regularize graph
  convolutional neural networks,'' in \emph{The 16th European Conference on
  Computer Vision}, 2020.

\bibitem{wang2019attributed}
C.~Wang, S.~Pan, R.~Hu, G.~Long, J.~Jiang, and C.~Zhang, ``Attributed graph
  clustering: A deep attentional embedding approach,'' in \emph{The
  Twenty-Eighth International Joint Conference on Artificial Intelligence},
  2019, pp. 3670--3676.

\bibitem{fan2020one2multi}
S.~Fan, X.~Wang, C.~Shi, E.~Lu, K.~Lin, and B.~Wang, ``One2multi graph
  autoencoder for multi-view graph clustering,'' in \emph{Proceedings of The
  Web Conference 2020}, 2020, pp. 3070--3076.

\bibitem{wang2020gcn}
X.~Wang, M.~Zhu, D.~Bo, P.~Cui, C.~Shi, and J.~Pei, ``Am-gcn: Adaptive
  multi-channel graph convolutional networks,'' in \emph{Proceedings of the
  26th ACM SIGKDD International Conference on Knowledge Discovery and Data
  Mining}, 2020, pp. 1243--1253.

\bibitem{chen2020learning}
X.~Chen, Y.~Zhang, I.~Tsang, and Y.~Pan, ``Learning robust node representation
  on graphs,'' \emph{arXiv preprint arXiv:2008.11416}, 2020.

\bibitem{liu2020self}
X.~Liu, F.~Zhang, Z.~Hou, Z.~Wang, L.~Mian, J.~Zhang, and J.~Tang,
  ``Self-supervised learning: Generative or contrastive,'' \emph{arXiv preprint
  arXiv:2006.08218}, 2020.

\bibitem{qiu2020gcc}
J.~Qiu, Q.~Chen, Y.~Dong, J.~Zhang, H.~Yang, M.~Ding, K.~Wang, and J.~Tang,
  ``Gcc: Graph contrastive coding for graph neural network pre-training,'' in
  \emph{Proceedings of the 26th ACM SIGKDD International Conference on
  Knowledge Discovery and Data Mining}, 2020, pp. 1150--1160.

\bibitem{yang2020understanding}
Z.~Yang, M.~Ding, C.~Zhou, H.~Yang, J.~Zhou, and J.~Tang, ``Understanding
  negative sampling in graph representation learning,'' in \emph{Proceedings of
  the 26th ACM SIGKDD International Conference on Knowledge Discovery and Data
  Mining}, 2020, pp. 1666--1676.

\bibitem{chen2020dynamic}
Y.~Chen, X.~Dai, M.~Liu, D.~Chen, L.~Yuan, and Z.~Liu, ``Dynamic relu,'' in
  \emph{The 16th European Conference on Computer Vision}, 2020.

\bibitem{ioffe2015batch}
S.~Ioffe and C.~Szegedy, ``Batch normalization: Accelerating deep network
  training by reducing internal covariate shift,'' \emph{arXiv preprint
  arXiv:1502.03167}, 2015.

\bibitem{newman2004finding}
M.~E. Newman and M.~Girvan, ``Finding and evaluating community structure in
  networks,'' \emph{Physical review E}, vol.~69, no.~2, p. 026113, 2004.

\bibitem{chen2020Asimple}
T.~Chen, S.~Kornblith, M.~Norouzi, and G.~Hinton, ``A simple framework for
  contrastive learning of visual representations,'' in \emph{Proceedings of the
  37th International Conference on Machine Learning}, 2020.

\bibitem{zhang2019prone}
J.~Zhang, Y.~Dong, Y.~Wang, J.~Tang, and M.~Ding, ``Prone: Fast and scalable
  network representation learning,'' in \emph{The 28th International Joint
  Conference on Artificial Intelligence}, 2019, pp. 4278--4284.

\bibitem{Zhou2020introduction}
Z.~Zhou, W.~Wang, W.~Gao, and L.~Zhang, \emph{Introduction to The Theory of
  Machine Learning}.\hskip 1em plus 0.5em minus 0.4em\relax China Machine
  Press, 2020.

\bibitem{ye2018deep}
F.~Ye, C.~Chen, and Z.~Zheng, ``Deep autoencoder-like nonnegative matrix
  factorization for community detection,'' in \emph{Proceedings of the 27th ACM
  International Conference on Information and Knowledge Management}, 2018, pp.
  1393--1402.

\bibitem{bojchevski2018deep}
A.~Bojchevski and S.~G¨¹nnemann, ``Deep gaussian embedding of graphs:
  Unsupervised inductive learning via ranking,'' in \emph{International
  Conference on Learning Representations}, 2018.

\bibitem{hollocou2019modularity}
A.~Hollocou, T.~Bonald, and M.~Lelarge, ``Modularity-based sparse soft graph
  clustering,'' in \emph{The 22nd International Conference on Artificial
  Intelligence and Statistics}, 2019, pp. 323--332.

\bibitem{velickovic2019deep}
P.~Velickovic, W.~Fedus, W.~L. Hamilton, P.~Li{\`o}, Y.~Bengio, and R.~D.
  Hjelm, ``Deep graph infomax,'' in \emph{International Conference on Learning
  Representations}, 2019.

\bibitem{mavromatis2020graph}
C.~Mavromatis and G.~Karypis, ``Graph infoclust: Leveraging cluster-level node
  information for unsupervised graph representation learning,'' \emph{arXiv
  preprint arXiv:2009.06946}, 2020.

\end{thebibliography}
\end{document}